%% file: it_Viterbi_v3.tex
\documentclass[onecolumn,draftcls]{IEEEtran}

\usepackage{latexsym}
\usepackage{amssymb}
\usepackage{bm}
\usepackage[dvips]{graphics}
\usepackage{graphicx}
\usepackage{psfrag}
\usepackage{amsfonts,amsmath,amssymb,hyperref}
\usepackage{algorithm}
\usepackage{algorithmic}
\usepackage{dsfont,color,subfigure}

\DeclareMathOperator*{\argmin}{arg\,min}

\input{defns.tex}

\newcommand{\Mc}{\mathcal{M}}

\newcommand{\La}{\Lambda}

\newtheorem{remark}{Remark}

\newtheorem{theorem}{Theorem}
\newtheorem{lemma}{Lemma}

\newtheorem{corollary}{Corollary}

\begin{document}

\title{Lossy  compression of discrete sources via Viterbi algorithm}

\author{Shirin Jalali, Andrea Montanari and Tsachy Weissman}
\maketitle

\newcommand{\p}{\mathds{P}}
\newcommand{\mb}{\mathbf{m}}
\newcommand{\bb}{\mathbf{b}}

\begin{abstract}

We present a new lossy compressor for discrete-valued sources. For coding a sequence $x^n$,  the encoder starts by assigning a certain cost to each possible  reconstruction sequence. It then finds the one that minimizes this cost and describes it losslessly to the decoder via a universal lossless compressor. The cost of each sequence is a linear combination of its distance from the sequence $x^n$ and a linear function of its $k^{\rm  th}$ order empirical distribution. The structure of the cost function allows the encoder to employ the  Viterbi algorithm to recover the minimizer of the cost. We identify a choice of the coefficients comprising the linear function of the empirical distribution used in the cost function which ensures that the algorithm universally achieves the optimum rate-distortion performance of any stationary ergodic source in the limit of large $n$, provided  that $k$ diverges  as $o(\log n)$. Iterative techniques for approximating the coefficients, which alleviate the computational burden of finding the optimal coefficients, are proposed and studied.   
\end{abstract}

\section{Introduction}
\input{intro}

\section{Conditional empirical entropy and its properties}\label{sec: conditional}
\input{cond_emp_entropy}


\section{Exhaustive search algorithm}\label{sec: exhaustive}
\input{search_alg}

\section{Linearized cost function}\label{sec: linearized cost}
\input{linearized_cost}

\section{Computing the coefficients}\label{sec: how to choose}
\input{choose_coeffs}

\section{Viterbi coder}\label{sec: Viterbi coding}
\input{viterbi_coder}

\section{Approximating the optimal coefficients}\label{sec: viterbi iterative}
\input{approx_coeffs}

\section{Simulation results}\label{sec: simulations Viterbi}
\input{sim_Viterbi}

\section{Conclusions}\label{sec: conclusion}

In this paper, a new approach to for fixed-slope lossy compression of discrete sources is proposed. The core ingredient is the use of the  Viterbi algorithm, which is a dynamic programing algorithm. It enables the encoder to find the reconstruction sequence with minimum cost. The encoder first assigns some weights to different contexts of length $k$, i.e, subsequences of length $k+1$, that appear within the reconstruction sequence. Then, the overall cost assigned to each possible reconstruction sequence is the sum of the weights of different contexts multiplied by their number of appearances in the sequence, plus some constant times the distance between the original sequence and the candidate reconstruction sequence.  From this definition, it turns out that the state of the Viterbi algorithm at time $t$ is the last $k$ symbols observed plus the current symbol in the sequence, i.e, $(y_{t-k},\ldots,y_{t})$. Therefore, the Trellis has overall $|\hat{\Xc}|^{k+1}$ different states, corresponding to $|\hat{\Xc}|^{k+1}$ different possible contexts of length $k$.  Hence for coding a sequence of length $n$, the computational complexity of the Viterbi algorithm will be of the order of $O(n2^{k+1})$.  We prove that there exists a set of optimal coefficients for which the described algorithm will achieve the rate-distortion performance for any stationary ergodic process. The problem is finding those weights. We provide an optimization problem whose solution can be used to find an asymptotically tight approximation of the optimal coefficients resulting in an overall scheme which is universal with respect to the class of stationary ergodic sources. However, solving this optimization problem is computationally demanding, and in fact infeasible in practice  for even moderate blocklengths. In order to overcome this problem, we propose an iterative approach for approximating the optimal coefficients. This approach is partially justified by a guarantee of convergance to at least a local minimum.

In the described iterative approach, the algorithm starts at a large slope (corresponding to a small distortion) and gradually decreases the slope until it hits the desired value. At each slope, the algorithm runs the Viterbi algorithm iteratively until it converges. An interesting possible next step is to explore whether there exisits a sequence of slopes converging to the desired value in a small number of steps (e.g. of $o(n)$) for which we can guarantee convergence of the algorithm to the global minimum at the end of the porcess. Existance of such sequence of slopes implies a universal lossy compression algorithm with moderate computatioal complexity. 

\renewcommand{\theequation}{A-\arabic{equation}}
\setcounter{equation}{0}  

\section*{APPENDIX A: Proof of Theorem \ref{thm: main Viterbi coeff}}  \label{app1} 

\input{appendix1}

\bibliographystyle{unsrt}
\bibliography{myrefs}

\end{document}

%% file: defns.tex
\usepackage{xspace}
\usepackage{bbm}

\newcommand{\Ac}{\mathcal{A}}

\newcommand{\Cc}{\mathcal{C}}

\newcommand{\Ec}{\mathcal{E}}

\newcommand{\Hc}{\mathcal{H}}

\newcommand{\Pc}{\mathcal{P}}

\newcommand{\Sc}{\mathcal{S}}

\newcommand{\Xc}{\mathcal{X}}
\newcommand{\Yc}{\mathcal{Y}}

\def\a{\alpha}
\def\b{\beta}

\def\d{\delta}
\def\e{\epsilon}

\def\l{\lambda}

\DeclareMathOperator\E{E}
\let\P\relax
\DeclareMathOperator\P{P}

\newcommand{\Bern}{\mathrm{Bern}}

\newcommand\ie{i.e.,\xspace}
\def\textiid{i.i.d.\@\xspace}
\newcommand\iid{\ifmmode\text{ i.i.d. } \else \textiid \fi}

\newcommand{\ind}{\mathbbmss{1}}


%% file: intro.tex
Consider the problem of universal lossy compression of stationary ergodic sources described as follows. Let $\mathbf{X}=\{X_i;\forall\; i\in\mathds{N}^{+}\}$ be a stochastic process and  let $\Xc$ denote its alphabet which is assumed  discrete and finite throughout this paper. Consider a family of source codes $\{\Cc_n\}_{n\geq 1}$. Each code $\Cc_n$ in this family consists of  an encoder $f_n$ and a decoder $g_n$ such that
\begin{align}
f_n:\Xc^n\to \{0,1\}^{*},
\end{align}
and
\begin{align}
g_n:\{0,1\}^{*}\to\hat{\Xc}^n,
\end{align}
where $\hat{\Xc}$ denotes the reconstruction alphabet which also is assumed to be finite and in most cases is equal to $\Xc$. $\{0,1\}^*$ denotes the set of all finite length binary sequences. The encoder $f_n$ maps each source block $X^n$ to a binary sequence of finite length, and the decoder $g_n$ maps the coded bits back to the signal space as $\hat{X}^n=g_n(f_n(X^n))$. Let  $l_n(f_n(X^n))$ denote the length of the binary sequence assigned to sequence $X^n$ by the encoder $f_n$. The performance of each code in this family is measured by the expected rate and the expected average distortion it induces. For a given source $\mathbf{X}$ and coding scheme $\Cc_n$, the expected rate $R_n$, and expected average distortion $D_n$, of  $\Cc_n$ in coding the process $\mathbf{X}$ are defined as follows:
\begin{align}
R_n =\E[\frac{1}{n}l_n(f_n(X^n))],
\end{align}
and
\begin{align}
D_n =\E[d_n(X^n,\hat{X^n})] \triangleq \E\left[\frac{1}{n}\sum\limits_{i=1}^n d(X_i,\hat{X}_i)\right],
\end{align}
where $\hat{X}^n=g_n(f_n(X^n))$, and $d:\Xc\times\hat{\Xc}\rightarrow{\mathds{R}}^+$ is a per-letter distortion measure.

For a given process and  any rate $R\geq0$, the minimum achievable distortion (cf.~\cite{cover} for exact definition of achievability) is characterized as
\cite{Shannon60}, \cite{Gallager}, \cite{book:Berger}
\begin{equation} \label{eq: rate-distortion function}
D(R,\mathbf{X})=\lim\limits_{n\rightarrow\infty}\min\limits_{p(\hat{X}^n|X^n):I(X^n;\hat{X}^n)\leq
R}\E[d_n(X^n,\hat{X}^n)].
\end{equation}
Similarly, for any distortion $D>0$, define $R(D,\mathbf{X})$ to denote the minimum required rate for achieving distortion $D$, i.e.,
\[
R(D,\mathbf{X}) = \min_{D(r,\mathbf{X})\leq D} r.
\]

Universal lossy compression codes are usually defined in the literature in one of the following modes \cite{YangZ:97}:
\begin{enumerate}
\item[I.] Fixed-rate: A family of lossy compression codes $\{\Cc_n\}$ is called fixed-rate universal, if for every stationary ergodic process $\mathbf{X}$, $R_n\leq R$, $\forall n\geq 1$, and
    \[
    \limsup_n D_n = D(R,\mathbf{X}).
    \]

\item[II.]  Fixed-distortion: A family of lossy compression codes $\{\Cc_n\}$ is called fixed-distortion universal, if for every stationary ergodic process $\mathbf{X}$, $D_n \leq D$, $\forall n\geq 1$, and
    \[
    \limsup_n R_n = R(D,\mathbf{X}).
    \]

\item[III.]  Fixed-slope: A family of lossy compression codes $\{\Cc_n\}$ is called fixed-slope universal, if there exists $\alpha>0$, such that for every stationary ergodic process $\mathbf{X}$
    \[
    \limsup_n [R_n+\alpha D_n] = \min\limits_{D\geq0}[R(D,\mathbf{X})+\alpha D].
    \]

\end{enumerate}

Existence of universal lossy compression codes for all these paradigms has already been established in the literature a long time ago \cite{Sakrison:70,Ziv:72,NeuhoffG:75,NeuhoffS:78,Ziv:80,GarciaN:82}.
The remaining challenging step is to design universal lossy compression algorithms that are implementable and appealing from  a practical viewpoint.

\subsection{Related prior work}

Unlike lossless compression, where there exists a number of well-known universal algorithms which are also attractive from a practical perspective (cf. Lempel-Ziv algorithm \cite{LZ} or arithmetic coding algorithm \cite{arith_coding}), in lossy compression, despite all the progress in recent years, no such algorithm is yet known. In this section, we briefly review some of the related literature on universal lossy compression with the main emphasis on the progress towards  the design of practically appealing algorithms.

There have been different approaches towards designing universal lossy compression algorithms. Among them the one with longest history is that of  tuning the well-known universal lossless compression algorithms to work for the lossy case as well. For instance, Cheung and Wei \cite{CheungW:90} extended the move-to-front transform  to the case where the reconstruction is not required to perfectly match the original sequence. One basic tool used in LZ-type compression algorithms, is the idea of string-matching, and hence there have been many attempts to find  optimal approximate string-matching. Morita and Kobayashi \cite{MoritaK:89} proposed a lossy version of LZW algorithm, and Steinberg and Gutman \cite{SteinbergG:93} suggested a fixed-database lossy compression algorithms based on string-matching. Although the extensions could all be implemented efficiently, they were later proved to be sub-optimal by Yang and Kieffer \cite{YangK:98}, even for memoryless sources. Another related example, is the work by Luczak and Szpankowski which proposes another suboptimal compression algorithm which again uses the ideas of approximate pattern matching \cite{LuczakS:97}. For some other related work see \cite{ZamirR:01} \cite{AtallahG:99}\cite{DemboK:99}.


Another well-studied approach to lossy compression is Trellis coded quantization \cite{MarcellinF:90} and more generally vector quantization (c.f.~\cite{BergerG:98}, \cite{book:GershoGray:92} and the references therein). Codes of this type are usually designed for a given distributions  encountered in a specific application.  For example, such codes are used in image compression (JPEG) or video compression (MPEG). Nevertheless, there have been attempts at extending such codes to more general settings. For instance Kasner,  Marcellin, and Hunt proposed  universal Trellis coded quantization which is used in the JPEG2000 standard \cite{KasnerM:99}.

There has been a lot of  progress in recent years in designing  \emph{non-universal} lossy compression algorithms of discrete memoryless sources.  Some examples of the recent work in this area are as follows. Wainwright and Maneva \cite{WainrightM:05} proposed a lossy compression algorithm based on message-passing ideas. The effectiveness of the scheme was shown by simulations.  Gupta and Verd\'u proposed an algorithm based on non-linear sparse-graph codes \cite{GuptaV:09}. Another algorithm with near linear complexity is suggested by Gupta, Verd\'u and Weissman in \cite{GuptaV:08}. The algorithm is  based on a `divide and conquer' strategy. It breaks the source sequence into sub-blocks and codes the subsequences  separately using a random codebook. Finally, the capacity-achieving polar codes proposed by Arikan \cite{Arikan:09_arxiv} for channel coding are shown to  be optimal for lossy compression of binary-symmetric memoryless sources in \cite{BabuR:09_arxiv}.

 The idea of fixed-slope universal lossy compression was first suggested by Yang, Zhang and Berger in \cite{YangZ:97}. They proposed a generic fixed-slope universal algorithm which leads to  specific coding algorithms based on different  universal lossless compression algorithms. Although the constructed algorithms are all universal, they involve computationally demanding minimizations, and hence are impractical. In \cite{YangZ:97}, the authors considered lowering the search complexity by choosing appropriate lossless codes which allow to replace the required exhaustive search by a low-complexity sequential search scheme that approximates the solution of the required minimization. However, these schemes only find an approximation of the optimal solution.

In a recent work \cite{JalaliW:09_arxiv}, a new implementable algorithm for fixed-slope lossy compression of discrete sources was proposed. Although the algorithm involves a minimization which resembles a specific realization of the generic cost proposed in \cite{YangZ:97}, it is  somewhat different. The reason is that the cost used in \cite{JalaliW:09_arxiv} cannot be derived directly from a lossless compression algorithm. The advantage of the new cost function is that it lends itself to rather naturally  Gibbs simulated annealing  in that the computational effort involved in each iteration is modest. It was shown that using a universal lossless compressor to describe the reconstruction sequence found by the annealing process to the decoder results in a scheme which is universal in the limit of many iterations and large block length. The drawback of the proposed scheme is that although its computational complexity per iteration is independent of the block length $n$ and linear in a parameter $k_n=o(\log n)$, there is no useful bound on the number of iterations required for convergence.

In this paper, motivated by the algorithm proposed in \cite{JalaliW:09_arxiv}, we propose another approach to fixed-slope lossy compression of discrete sources. We start by making a linear approximation of the cost used in \cite{JalaliW:09_arxiv}. The cost assigned to each possible reconstruction sequence consists of a linear   combination of two terms:  a linear function of its empirical distribution plus its distance to  (distortion from) the source sequence. We show that there exists proper coefficients such that minimizing the linearized cost function results in the same performance as  would minimizing the original cost. The advantage of the modified cost  is that its minimizer can be found simply using the Viterbi algorithm.

\subsection{Organization of this paper}

The organization of the paper is as follows. In Section \ref{sec: conditional}, the count matrix of a sequence and its empirical conditional entropy is introduced and some of their properties are  studied. Section \ref{sec: exhaustive} reviews the fixed-slope universal lossy compression algorithm used in \cite{JalaliW:09_arxiv}. Section \ref{sec: linearized cost} describes a new coding scheme for fixed-slope lossy compression derived by replacing part of the cost used in the mentioned exhaustive-search algorithm by a linear function. We prove that using appropriate coefficients for the linear function, the performance of the two algorithms remains the same. In Section \ref{sec: how to choose}, a method for approximating these optimal coefficients is presented. This method,  along with the result of the previous section, gives rise to  a fixed-slope universal lossy compression algorithm that achieves the rate-distortion performance for any discrete stationary ergodic source. The advantage of this modified cost is discussed in Section \ref{sec: Viterbi coding} where we show that the minimizer of the new cost can be found using the Viterbi algorithm. The method introduced for approximating the coefficients is computationally demanding, and hence is impractical. Therefore, in Section \ref{sec: viterbi iterative}, we discuss a low-complexity iterative detour for approximating the coefficients. Section \ref{sec: simulations Viterbi} presents some simulations results and, finally, Section \ref{sec: conclusion} concludes the paper with a discussion of some future directions.

%% file: cond_emp_entropy.tex
For any $y^n\in\Yc^n$, let the $|\Yc|\times|\Yc|^{k}$ matrix $\mathbf{m}(y^n)$ denote its $(k+1)^{\rm th}$ order empirical distribution\footnote{For any set $\Ac$, $|\Ac|$ denotes its size.}. For $\bb=(b_1,\ldots,b_k)\in\Yc^k$, and $\b\in\Yc$, the element in the $\b^{\rm th}$ row and the $\bb^{\rm th}$ column of the matrix $\mb$, $m_{\b,\bb}$, is defined as
\begin{align}\label{eq: empirical count matrix}
m_{\b,\bb}(y^n) \triangleq \frac{1}{n} \left| \left\{1 \leq i \leq n: y_{i-k}^{i-1} = \bb, y_i=\b]    \right\}\right|,
\end{align}
where here and throughout the paper we assume a cyclic convention whereby $y_{i}=y_{i+n}$ for $i\leq 0$. 

Based on the distribution induced by $\mb(y^n)$, define the $k^{\rm th}$ order conditional empirical entropy of $y^n$, $H_k(y^n)$, as
\begin{equation}\label{eq: emp cond entropy}
   H_k (y^n) \triangleq H(Z_{k+1}|Z^{k}),
\end{equation}
where $Z^{k+1}$ is assumed to be distributed according to $\mb$, i.e., 
\begin{equation}\label{eq: empirical distribution}
   \P \left(Z^{k+1} = [b_1,\ldots,b_k,\b]=[\bb,\b]\right) = m_{\b,\bb}(y^n).
\end{equation}
For a vector
$\mathbf{v}= (v_1, \ldots , v_\ell)^T$ with non-negative
components, we let $\mathcal{H}(\mathbf{v})$ denote the entropy
of the random variable whose probability mass function (pmf) is
proportional to $\mathbf{v}$. Formally,
\begin{equation}\label{eq: entropy_vec}
\mathcal{H} (\mathbf{v}) = \left\{ \begin{array}{cc}
                           \sum\limits_{i=1}^\ell \frac{v_i}{\| \mathbf{v}
\|_1}  \log \frac{\| \mathbf{v} \|_1}{v_i} &  \mbox{ if }  \mathbf{v}
\neq (0, \ldots , 0)^T \\
                           0 & \mbox{ if } \mathbf{v}  = (0, \ldots , 0)^T,
                         \end{array}
\right.
\end{equation}
where $0\log(0) = 0$ by convention.
With this notation, the conditional empirical entropy $H_k(y^n)$ defined in \eqref{eq: emp cond entropy} is readily seen to be expressible  in terms of $\mb(y^n)$ as
\begin{equation}\label{eq: alternative representation of Hk}
H_k (y^n) \triangleq H(\mb(y^n))  \triangleq \sum_{\bb} \mathcal{H} \left(
\mb_{\cdot,\bb} \right)\sum_{\b\in\Yc}m_{\b,\bb},
\end{equation}
where  $\mb_{\cdot,\bb}$ denotes the column of $\mb$ indexed by $\bb$. \\
\begin{remark} Note that $H_k(\cdot)$ has a discrete domain, while the domain of $H(\cdot)$ is continuous and consists of all $|\Yc|\times|\Yc|^k$ matrices with positive real entries adding up to one.  In other words,
\begin{align}
H_k:\Yc^{n}\to [0,1],
\end{align}
but
\begin{align}
H:[0,1]^{|\Yc|}\times[0,1]^{|\Yc|^k}\to [0,1].
\end{align}
\end{remark}

%
%
%

Conditional empirical entropy of sequences, $H_k(\cdot)$, plays key role in our results. Hence, in the following two subsections, we focus on this function, and study some of its properties.

\subsection{Concavity}
We prove that like the standard entropy function, conditional empirical entropy is also a concave function. By
definition
\begin{align}
H(\mb)=\sum\limits_{\bb\in\Yc^k}(\sum_{\b\in\Yc}m_{\b,\bb})\Hc(\mb_{\cdot,\bb}),
\end{align}
where $\Hc(\cdot)$ is defined in \eqref{eq: entropy_vec}. We need to show that for any $\theta\in[0,1]$,
and matrices $\mb^{(1)}$ and $\mb^{(2)}$  with non-negative components adding up to one,
\begin{align}
\theta H(\mb^{(1)})+\bar{\theta}H(\mb^{(2)})\leq
H(\theta\mb^{(1)}+\bar{\theta}\mb^{(2)}),
\end{align}
where $\bar{\theta}=1-\theta$. From the concavity of entropy function $\Hc$, it follows that
\begin{align}
&\theta(\sum_{\b\in\Yc} m^{(1)}_{\b,\bb})\Hc(\mb^{(1)}_{\cdot,\bb})+\bar{\theta}(\sum_{\b\in\Yc} m^{(2)}_{\b,\bb})\Hc(\mb^{(2)}_{\cdot,\bb})\nonumber\\
&=(\theta(\sum_{\b\in\Yc} m^{(1)}_{\b,\bb})+\bar{\theta}
(\sum_{\b\in\Yc} m^{(2)}_{\b,\bb}))\sum\limits_{i\in\{1,2\}}\frac{\theta_i(\sum_{\b\in\Yc} m^{(i)}_{\b,\bb})}{(\theta(\sum_{\b\in\Yc} m^{(1)}_{\b,\bb})+\bar{\theta}
(\sum_{\b\in\Yc} m^{(2)}_{\b,\bb}))}\Hc(\mb_{\cdot,\bb}^{(i)})\nonumber\\
&\leq (\theta(\sum_{\b\in\Yc} m^{(1)}_{\b,\bb})+\bar{\theta}
(\sum_{\b\in\Yc} m^{(2)}_{\b,\bb}))\Hc(\theta \mb_{\cdot,\bb}^{(1)} + \bar{\theta}\mb_{\cdot,\bb}^{(2)} ),\label{eq: concavity}
\end{align}
where $\theta_1\triangleq1-\theta_2\triangleq\theta$. Summing up both sides of \eqref{eq: concavity} over all $\bb\in\Yc^{k}$ yields the desired result.

\subsection{Stationarity  condition}\label{subsec:stationary}
Let  $p(y^{k+1})$ be a given pmf defined on $\Yc^{k+1}$. Under what condition(s) does there exist a a stationary process with its $(k+1)^{\rm th}$ order distribution equal to $p$?
\begin{lemma}\label{lemma: stationary}
The necessary and sufficient condition for $\{p(y^{k+1})\}_{y^{k+1}\in\Yc^{k+1}}$ to represent the $(k+1)^{\rm th}$ order marginal distribution of a stationary process is 
\begin{align}
\sum\limits_{\b\in\Yc}p(\b, y^k)=\sum\limits_{\b\in\Yc}p(y^k,\b),\; \forall\;y^k \in \mathcal{Y}^k. \label{eq: stat cond}
\end{align}
\end{lemma}
\begin{proof}
\begin{itemize}
\item[i.] Necessity: The necessity of \eqref{eq: stat cond}
    is just a direct result of the stationarity of the process. If $p(y^{k+1})$ is to
    represent the $(k+1)^{\rm th}$ order marginal
    distribution of a stationary process $\mathbf{Y}=\{Y_i\}$, then it should be
    consistent with the $k^{\rm th}$ order marginal
    distribution. Hence,  \eqref{eq: stat cond} should hold.
    
\item[ii.] Sufficiency: In order to prove the sufficiency, we
    assume that \eqref{eq: stat cond} holds, and build a
    stationary process with  $(k+1)^{\rm th}$ order
    marginal distribution equal to $p(y^{k+1})$. Let $\mathbf{Y}=\{Y_i\}_i$ be a
    Markov chain of order $k$ whose transition probabilities are defined as
\begin{align}
\P(Y_{k+1}=y_{k+1}|Y^k=y^k) \triangleq q(y_{k+1}|y^k) \triangleq \frac{p(y^{k+1})}{p(y^k)},
\end{align}
where \[ p(y^k)\triangleq \sum\limits_{\b\in\Yc}p(\b, y^k)=\sum\limits_{\b\in\Yc}p(y^k,\b).\] Now, given \eqref{eq: stat cond}, it is easy to check that $p(y^{k+1})$ is the $(k+1)^{\rm  th}$ order stationary distribution of the defined Markov chain. Therefore, $\mathbf{Y}$ is a stationary process with the desired marginal distribution.


\end{itemize}

\end{proof}

Throughout the paper, we refer to the condition stated in \eqref{eq: stat cond} as the \emph{stationarity condition}.

\begin{corollary}
For any $|\Yc|\times|\Yc|^{k}$ matrix $\mb$ corresponding to the $(k+1)^{\rm th}$ order empirical distribution of some $y^n\in\Yc^n$, there exists a stationary process whose marginal distribution coincides with $\mb$.
\end{corollary}
\begin{proof}
From Lemma \ref{lemma: stationary}, we only need to show that \eqref{eq: stat cond} holds, i.e.,
\begin{align}
\sum\limits_{\b\in\Yc}m_{\b,\bb}=\sum\limits_{\b\in\Yc}m_{b_k,[\b,b_1\ldots,b_{k-1}]}, \forall \;\bb\in\Yc^k, \label{eq: m stat}
\end{align}
which obviously holds  because both sides of \eqref{eq: m stat} are equal to $|\{i:y_{i+1}^{i+k}=\bb\}|/(n-k)$.
\end{proof}

%% file: search_alg.tex
Consider the following lossy source coding algorithm. Given $\a>0$, for encoding sequence $x^n\in\Xc^n$, find
\begin{align}
\hat{x}^n=\argmin\limits_{y^n\in\hat{\Xc}^n}[H_k(y^n)+\a d_n(x^n,y^n)],\label{eq: exhaustive_search}
\end{align}
and describe $\hat{x}^n$ using the Lempel-Ziv coding algorithm. As proved before \cite{YangZ:97}, \cite{JalaliW:09_arxiv}, the described algorithm is a universal lossy  compression algorithm. That is, for any stationary ergodic source $\mathbf{X}$,
\begin{align}
\frac{1}{n}\ell_{\footnotesize{\rm LZ}}(\hat{X}^n)+\alpha d_n(X^n,\hat{X}^n)\to \min[R(D,\mathbf{X})+\a D],\;\;{\rm a.s.},
\end{align}
where  $X^n$ is  generated by the source $\mathbf{X}$, and $\hat{X}^n$ denotes the minimizer of \eqref{eq: exhaustive_search} for the input $X^n$. Here $\ell_{\footnotesize{\rm LZ}}$ denotes the length of the codeword assigned to $\hat{X}^n$ by the Lempel-Ziv algorithm \cite{LZ}.  Clearly, given the size of the search space, this is not an implementable algorithm. An approach for approximating the solution of \eqref{eq: exhaustive_search} using Markov chain Monte Carlo methods has been suggested in  \cite{JalaliW:09_arxiv}. One problem with the MCMC-based algorithms is that no useful bound is yet known on the required number of iterations. Moreover, the performance of the algorithm depends on the cooling process chosen. There exist cooling schedules with guaranteed convergence, but they are very slow, and usually not used in practice. On the other hand, if we use faster cooling processes, there is a risk of getting stuck in a local minima and missing the optimum solution. The goal of this paper is to propose a new approach for approximating the solution of \eqref{eq: exhaustive_search}. This new approach, as we show later,  suggests a new implementable algorithm for lossy compression. The main idea here is using linear approximation of the conditional entropy function, $H(\mb)$,  at some point $\mb_0$, and proving that if $\mb_0$ is chosen correctly, then while we have reduced the exhaustive search algorithm to the Viterbi algorithm, we have not changed its performance.

%% file: linearized_cost.tex
Consider the problems (P1) and (P2) described  by \eqref{eq:P1} and \eqref{eq:P2} respectively, where (P1) corresponds to the  optimization required by the exhaustive search lossy compression scheme described in \eqref{eq: exhaustive_search}, and (P2) involves a similar optimization problem. The difference between (P1) and (P2) is that the term corresponding to conditional empirical entropy in (P1), which is a highly non-linear function of $\mb$, is replaced by a linear function of $\mb$.
\begin{align}
(\textmd{P1}):\quad\min\limits_{y^n}\;\;\left[H(\mb(y^n))+\a d_n(x^n,y^n)\right],\label{eq:P1}
\end{align}
and
\begin{align}
(\textmd{P2}):\quad\min\limits_{y^n}\;\;\left[\sum\limits_{\b}\sum\limits_{\bb} \l_{\b,\bb}m_{\b,\bb}(y^n)+\a  d_n(x^n,y^n)\right],\label{eq:P2}
\end{align}
where $\{\l_{\b,\bb}\}_{\b,\bb}$ are a set of real-valued coefficients. In this section we are interested in answering the following question:\\
Is it possible to choose the set of coefficients $\{\l_{\b,\bb}\}_{\b,\bb}$,  $\b\in\hat{\Xc}$ and $\bb\in\hat{\Xc}^k$, such that (P1) and (P2) have the same set of minimizers, or at least the set of minimizers of (P2) is a subset of the minimizers of (P1)? \\
The reason we are interested in answering this question is that if the answer is affirmative, then instead of solving (P1) one can solve (P2), which we describe in Section \ref{sec: Viterbi coding} can be done efficiently via the Viterbi algorithm.

Let $\Sc_1$ and $\Sc_2$ denote the set of minimizers of (P1) and (P2) respectively. Consider some $z^n\in \Sc_1$, and let
$\mb^*_n=\mb(z^n)$, and let the  coefficients used  in (P2) 
\begin{align}
\l_{\b,\bb}&=\left.\frac{\partial}{\partial m_{\b,\bb}}H(\mb)\right|_{\mb^*_n}\nonumber\\ 
&=\log({\sum_{\b'}m^*_{\b',\bb} \over m^*_{\b,\bb}}).\label{eq: def of lambda}
\end{align}

\begin{theorem}\label{thm:S1_S2}
If the coefficients used in (P2) are chosen according to \eqref{eq: def of lambda}, then the minimum values of  (P1) and (P2) will be the same. Moreover,
\[\Sc_2 \subset \Sc_1\]
and contains all the sequences $w^n\in\Sc_1$ with $\mb(w^n)=\mb_n^*$.
\end{theorem}
\begin{proof}
Since, as proved earlier,  $H(\mb)$ is concave in $\mb$, for any empirical
count matrix $\mb$, we have
\begin{align}
H(\mb) &\leq H(\mb^*) + \sum\limits_{\b,\bb}
\left.\frac{\partial}{\partial m_{\b,\bb}}H(\mb)\right|_{\mb^*_n}(m_{\b,\bb}-m^*_{\b,\bb})\\
 &= H(\mb^*) + \sum\limits_{\b,\bb}
\l_{\b,\bb}(m_{\b,\bb}-m^*_{\b,\bb})\\
&\triangleq \hat{H}(\mb). \label{eq:def of H_hat}
\end{align}
Adding a constant to the both sides of \eqref{eq:def of H_hat},  we conclude that for any $y^n\in\hat{\Xc}^n$,
\begin{align}
H(\mb(y^n))+\a d_n(x^n,y^n)\leq \hat{H}(\mb(y^n))+\a d_n(x^n,y^n).\label{eq: t1}
\end{align}
Taking the minimum of both sides of \eqref{eq: t1} yields
\begin{align}
\min\limits_{y^n}[H(\mb(y^n))+\a d_n(x^n,y^n)]&\leq
\min\limits_{y^n}[\hat{H}(\mb(y^n))+\a d_n(x^n,y^n)]\\
&\leq \hat{H}(\mb(z^n))+\a d_n(x^n,z^n)\\
&= H(\mb(z^n))+\a d_n(x^n,z^n)\\
&=\min\limits_{y^n}[H(\mb(y^n))+\a d_n(x^n,y^n)],
\end{align}
because $z^n\in\Sc_1$. Therefore,
\begin{align}
\min\limits_{y^n}[H(\mb(y^n))+\a d_n(x^n,y^n)] =\min\limits_{y^n}[\hat{H}(\mb(y^n))+\a d_n(x^n,y^n)],
\end{align}
i.e., (P1) and (P2) have the same minimum values.

For any sequence $w^n$ with $\mb(w^n)\neq \mb^*_n$, by strict concavity of $H(\mb)$,
\begin{align}
\hat{H}(\mb(w^n))+\a d_n(x^n,w^n)&>H(\mb(w^n))+\a d_n(x^n,w^n),\\
&\geq\min_{y^n} [H_k(y^n)+\a d_n(x^n,y^n)].
\end{align}
Hence, the empirical count matrices of all the sequences in $\Sc_2$, i.e., all the minimizers of (P2) for the selected coefficients, are equal to $\mb^*_n$.

Let $w^n\in\Sc_2$. We prove that $w^n\in\Sc_1$ as well. As we just proved, $\mb(w^n)=\mb(z^n)=\mb_n^*$.  Moreover, since both $z^n$ and $w^n$ belong to $\Sc_2$,
\begin{align}
\min_{y^n}[\hat{H}(\mb(y^n))+\a d_n(x^n,y^n)]&=\hat{H}(\mb(w^n))+\a d_n(x^n,w^n)\nonumber\\
&=\hat{H}(\mb(z^n))+\a d_n(x^n,z^n).
\end{align}
Therefore, $d_n(x^n,w^n)=d_n(x^n,z^n)$, and consequently,
\begin{align}
H_k(w^n)+\a d_n(w^n,x^n)&=H_k(z^n)+\a d_n(z^n,x^n),\nonumber\\
&=\min_{y^n}[H_k(y^n)+\a d_n(y^n,x^n)],
\end{align}
which proves that $w^n\in\Sc_1$, and concludes the proof.
\end{proof}

Theorem \ref{thm:S1_S2} states that if the optimal type $\mb^*_n$ is known, then  the desired coefficients can be computed according to \eqref{eq: def of lambda}, and solving (P2) instead of (P1) using the computed coefficients finds a minimizer of (P1). In Section \ref{sec: Viterbi coding}, we describe  how (P2) can be solved efficiently using Viterbi algorithm for a given set of coefficients.  The problem of course is that the optimal type $\mb^*_n$ required for computing the desired coefficients is not known to the encoder (since knowledge of $\mb_n^*$ seems to require solving (P1) which is the problem we are trying to avoid). In Section \ref{sec: how to choose}, we introduce  another optimization problem whose solution is a good approximation of  $\mb^*_n$, and hence of the desired coefficients $\{\l_{\b,\bb}\}$  when substituting in \eqref{eq: def of lambda}.

%% file: choose_coeffs.tex
As mentioned in the previous section, there exists a set of coefficients for which (P1) and (P2) have the same value. However, computing the desired coefficients requires the knowledge of $\mb_n^*$ which is not available without solving (P1). In order to alleviate this issue, in this section we introduce another optimization problem that gives an asymptotically tight approximation of $\mb_n^*$, and therefore a reasonable approximation of the set of coefficients. 

For a given sequence $x^n$ and a given order $k$, let $\Mc^{(k)}=\Mc^{(k)}(x^n)$  be the set of all jointly stationary probability distributions on $(X^{k},\hat{X}^{k})$ (in the sense of Lemma \ref{lemma: stationary}) such that their marginal distributions with respect to $X$ coincide with the $k^{\rm{th}}$ order empirical distribution induced by $x^n$ defined as follows
\begin{align} \label{eq: emp dist}
\hat{p}_{[x^n]}^{(k)}(a^{k})&\triangleq\frac{|\{1\leq i\leq n:(x_{i-k},\ldots,x_{i-1})=a^{k}\}|}{n},\nonumber\\
                                &=\frac{1}{n}\sum\limits_{i=1}^n\ind_{x_{i-k}^{i-1}=a^{k}},
\end{align}
where $a^{k}\in\Xc^{k}$. More specifically a distribution $p^{(k)}$ in $\Mc^{(k)}$ should satisfy the following  two constraints:
\begin{enumerate}
\item Stationarity condition: as described in Section \ref{subsec:stationary}, for any $a^{k-1}\in\Xc^{k-1}$ and $b^{k-1}\in\hat{\Xc}^{k-1}$,
\begin{align}
\sum\limits_{a_k\in\Xc,b_k\in\hat{\Xc}} p^{(k)}(a^k,b^k)=\sum\limits_{a_k\in\Xc,b_k\in\hat{\Xc}} p^{(k)}(a_ka^{k-1},b_kb^{k-1}).
\end{align}
\item Consistency: for each $a^k\in\Xc^k$,
\begin{align}
\sum\limits_{b^k\in\hat{\Xc}^k} p^{(k)}(a^k,b^k)=\hat{p}_{[x^n]}^{(k)}(a^{k}).
\end{align}
\end{enumerate}

For given $x^n$, $k$ and $\ell>k$, consider the following optimization problem
\begin{align}
\min&\quad H(\hat{X}_{k+1}|\hat{X}^k)+\a\E d(X_1,\hat{X}_1)\nonumber\\
\textmd{s.t.}&\hspace{4mm} (X^{\ell},\hat{X}^{\ell})\sim p^{(\ell)}\nonumber\\
             &\hspace{4mm} p^{(\ell)}\in\Mc^{(\ell)}.\label{eq: original}
\end{align}

\begin{remark} Note that  the rate-distortion function of a stationary ergodic process $\mathbf{X}$ has the following representation  \cite{GrayN:75}:
\begin{align}
R(D,\mathbf{X})&=\inf\{\bar{H}(\mathbf{\hat{X}}): \;(\mathbf{X},\mathbf{\hat{X}})\;\; {\rm jointly} \; {\rm stationry}\; {\rm and}\; {\rm ergodic}, \; {\rm and }\; \E d(X_0,\hat{X}_0)\leq D\}, \nonumber\\
&=\inf_{k \geq 1} \inf \{ H(\hat{X}_{k+1}|\hat{X}^k): (\mathbf{X}, \hat{\mathbf{X}}) \;{\rm jointly}\;{\rm stationary}\;{\rm and}\;{\rm ergodic,}\; {\rm and}\; \E d(X_0, \hat{X}_0) \leq D \}, \label{eq: R_D_alter}
\end{align}
where $\bar{H}(\hat{\bold{X}})$ denotes the entropy rate of the stationary ergodic process $\hat{\bold{X}}$, \ie
\begin{align}
\bar{H}(\hat{\mathbf{X}}) \triangleq\lim\limits_{n\to\infty}H(\hat{X}_{n+1}|\hat{X}^n).\label{eq: entropy_rate}
\end{align}

This representation gives the motivating intuition behind the optimization described in \eqref{eq: original}. It shows  that \eqref{eq: original} is basically performing  the search required by \eqref{eq: R_D_alter}.
\end{remark}

Using the properties  of the set $\Mc^{(\ell)}$, and the definition of conditional empirical entropy,  \eqref{eq: original} can be written more explicitly as
\begin{align}
\min\quad H(\mb)&+\a \sum\limits_{a\in\Xc,b\in\hat{\Xc}} d(a,b)q(a,b)\nonumber\\
\textmd{s.t.}\hspace{10mm} & 0\leq  p^{(\ell)}(a^{\ell},b^{\ell})\leq 1,\quad\forall\;\a^{\ell}\in\Xc^{\ell},b^{\ell}\in\hat{\Xc}^{\ell},\nonumber\\
& \sum_{a^{\ell},b^{\ell}}p^{(\ell)}(a^{\ell},b^{\ell})=1,\quad \forall \; a^{\ell}\in\Xc^{\ell},b^{\ell}\in\hat{\Xc^{\ell}},\nonumber\\
& \sum\limits_{a_{\ell}\in\Xc,b_{\ell}\in\hat{\Xc}} p^{(\ell)}(a^{\ell},b^{\ell})=\sum\limits_{a_{\ell}\in\Xc,b_{\ell}\in\hat{\Xc}} p^{({\ell})}(a_{\ell}a^{{\ell}-1},b_kb^{{\ell}-1}),\nonumber\\
&\hspace{3.8cm}\quad\forall\;
a^{\ell-1}\in\Xc^{\ell-1},b^{\ell-1}\in\hat{\Xc}^{\ell-1},\nonumber\\
&\sum\limits_{b^{\ell}\in\hat{\Xc}^{\ell}} p^{(\ell)}(a^{\ell},b^{\ell})=\hat{p}_{[x^n]}^{(\ell)}(a^{\ell})\quad\forall\; a^{\ell}\in\Xc^{\ell},\nonumber\\
&q(a,b)=\sum\limits_{a^{\ell-1}\in\Xc^{\ell-1},b^{\ell-1}\in\hat{\Xc}^{\ell-1}}p^{(\ell)}(aa^{\ell-1},bb^{\ell-1})\nonumber\\
& m_{\b,\bb}=\sum\limits_{a^{\ell}\in\Xc^{\ell},b^{\ell-k}\in\hat{\Xc}^{\ell-k}}p^{(\ell)}(a^{\ell},\bb\b b^{\ell-k}),\quad\forall\;\b,\bb.\label{eq: optimization}
\end{align}

Note that the optimization in \eqref{eq: optimization} is done over the joint distributions $p^{(\ell)}$ of $(X^{\ell},\hat{X}^{\ell})$. Let $\hat{\Pc}_n^{*}$ denote the set of minimizers of \eqref{eq: optimization}, and $\hat{\Sc}_n^*$ be their $(k+1)^{\rm th}$ order marginalized versions with respect to $\hat{X}$. Let $\{\hat{\l}_{\b,\bb}\}_{\b,\bb}$ be the coefficients evaluated at some $\hat{\mb}_n^*\in\hat{\Sc}_n^*$ using \eqref{eq: def of lambda}. Let $\mathbf{X}$ be a stationary ergodic source, and $R(\mathbf{X}, D)$ denote its rate distortion function. Finally, let $\hat{X}^n$ be the reconstruction sequence obtained by solving (P2) (recall \eqref{eq:P2}) at the evaluated coefficients.

\begin{theorem}\label{thm: main Viterbi coeff}
If $k = k_n = o(\log n)$, $\ell = \ell_n = o(n^{1/4})$ and  $k=o(\ell)$ such that $k_n, \ell_n\to\infty$, as $n\to\infty$, then  for any stationary ergodic source 
\begin{equation}\label{eq: achieving optimal point on the rd curve}
 H_{k}(\hat{X}^n) +\a  d_n (X^n, \hat{X}^n )  \stackrel{n \rightarrow \infty}{\longrightarrow}  \min_{D \geq 0} \left[ R(\mathbf{X}, D)
    +\a D \right],\;\;{\rm a.s.}
\end{equation}

\end{theorem}

The proof of Theorem \ref{thm: main Viterbi coeff} is presented in Appendix A.

\begin{remark} Theorem \ref{thm: main Viterbi coeff} implies the fixed-slope universality of the scheme which does the lossless compression of the reconstruction by first describing its count matrix (costing a number of bits which is negligible for large n) and then doing the conditional entropy coding. 
\end{remark}

\begin{remark}
Note that all the constraints in \eqref{eq: optimization} are linear, and the cost is a concave function. Hence, overall, we have a concave minimization problem 
(of dimension $|\Xc|^{\ell}|\hat{\Xc}|^{\ell}+|\hat{\Xc}|^{k+1}+|\Xc| |\hat{\Xc}|$). \textcolor{white}{aaa}

\end{remark}

%% file: viterbi_coder.tex
In this section, we show how, for a given set of coefficients, $\{\l_{\bb,\b}\}$, (P2) can be solved efficiently via the Viterbi algorithm \cite{Viterbi:67}, \cite{Forney:73}.

Note that the linearized cost used in (P2) can also be written as
\begin{align}
&\sum_{\substack{\bb\in\hat{\Xc}^k\\\b\in\hat{\Xc}}} \left[\l_{\b,\bb}m_{\b,\bb}(y^n)+\a d_n(x^n,y^n)\right]=\frac{1}{n}\sum\limits_{i=1}^n\left[
\l_{y_i,y_{i-k}^{i-1}}+\a d(x_i,y_i)\right].\label{eq:3}
\end{align}
The advantage of this alternative representation is that, as we will describe, instead of using simulated annealing, we can find the
sequence that exactly minimizes \eqref{eq:3} via the Viterbi algorithm, which is a dynamic programming optimization method
for finding the path of minimum weight in a Trellis diagram efficiently.  For $i=k+1,\ldots,n$, let
\begin{align}
s_i\triangleq y_{i-k}^i\
\end{align}
to be the state at time $i$, and define $\cal{S}$ to be the set of all $|\hat{\Xc}|^{k+1}$
possible states.  From this definition, the state at time $i$, $s_i$, is determined by the
state at time $i-1$, $s_{i-1}$, and $y_i$. In other words, $s_{i}=g(s_{i-1},y_i)$, for some
\[g:\Sc\times\hat{\Xc}\to\Sc.\]
This representation leads to a Trellis diagram corresponding to the evolution of the states
$\{s_i\}_{i=k+1}^n$ in which each state has $|\hat{\Xc}|$
states leading to it and $|\hat{\Xc}|$ states branching from
it. To the edge $e=(s',s)$ connecting states $s'$ and $s=b^{k+1}$ at stage $i$, we assign the weight $w_i(e)$ defined as
\begin{align}
w_i(e):=\l_{b_{k+1},b^k}+\a d(x_i,b_{k+1}).\label{eq: def w}
\end{align}
 In this representation, there is a 1-to-1 correspondence between sequences $y^n\in\hat{\Xc}^n$, and
sequences of states $\{s_i\}_{i=k+1}^n$, and minimizing \eqref{eq:3} is equivalent to finding the path of minimum
weight in the corresponding Trellis diagram, i.e., the path $\{s_i\}_{i=k+1}^n$ that minimizes $\sum_{i=k+1}^n w_i(e_i)$, where $e_i=(s_{i-1},s_i)$. Solving this minimization can readily be done by the Viterbi
algorithm which can be described as follows. For each state $s$, let $\mathcal{L}(s)$ be the $|\hat{\Xc}|$ states leading to it, and for any $i>1$, define
\begin{align}
C_i(s):=\min\limits_{s'\in\mathcal{L}(s)}[w_i((s',s))+C_{i-1}(s')].
\end{align}
For $i=1$ and $s=b^{k+1}$, let $C_1(s):=\l_{b_{k+1},b^k}+\a
d_{k+1}(x^{k+1},b^{k+1})$. Using this procedure, each state $s$ at
each time $j$ has a path of length $j-k-1$ which is the minimum path
among all the possible paths between the states from time $i=k+1$ to
$i=j$ such that $s_j=s$. After computing
$\{C_i(s)\}$ for all $s\in\Sc$ and all $i\in\{k+1,\ldots,n\}$, at time
$i=n$, let
\begin{align}
s^*=\argmin_{s\in\Sc}C_n(s).
\end{align}
It is not hard to see  that the path leading to $s^*$ is the
path of minimum weight among all possible paths.

Note that the computational complexity of this procedure is linear in $n$ but exponential in $k$ because the number of states increases exponentially with $k$. Therefore, given the coefficients $\{\l_{\bb,\b}\}$, solving
(P2) is straightforward using the Viterbi algorithm. The problem is
finding an approximation of the optimal coefficients. The
procedure outlined in Section \ref{sec: linearized cost} for
finding the coefficients involves solving a concave minimization problem  of dimension that becomes intractable even for moderate values of $n$. To bypass this process, an alternative heuristic
method is proposed in the next section. The effectiveness of this approach is discussed in the next
section through some simulations.

%% file: approx_coeffs.tex
As we discussed in Section \ref{sec: linearized cost}, having known the optimal coefficients, solving (P2) which can be done using the Viterbi algorithm is equivalent to solving (P1) which has exponential complexity in $n$. However, the problem is finding such desired coefficients. In Section \ref{sec: how to choose}, it was proposed that for finding a good approximation of these coefficients, one method is to solve \eqref{eq: optimization} and find $\hat{\mb}^{*}$. Then an approximation of the coefficients $\{\l_{\b,\bb}\}$ can be made via \eqref{eq: def of lambda}  by evaluating the partial derivatives of $H(\mb)$ at $\hat{\mb}^{*}$. But solving \eqref{eq: optimization} requires solving a concave minimization problem of dimension which is demanding for even moderate values of $n$. Therefore, in this section, we consider a detour with moderate computational complexity.

First, assume that the desired distortion is small, or equivalently $\a$ is large. In that case, the distance between the original sequence $x^n$ and its quantized version $\hat{x}^n$ should be small. Therefore, their types, i.e., their $(k+1)^{\rm th}$ order empirical distributions, are close. Hence, the  coefficients computed based on $\mb(x^n)$ provide a reasonable approximation of the coefficients derived from $\mb^*$. This implies that if our desired distortion is small, one possibility is to compute the type of the input sequence, and evaluate the coefficients at $\mb(x^n)$.

In the case where the desired distortion is not very small, we can use an iterative approach as follows. Start with $\mb(x^n)$. Compute the coefficients from \eqref{eq: def of lambda} at $\mb(x^n)$. Employ Viterbi algorithm to solve (P2) at the computed coefficients. Let $\hat{x}^n$ denote the output sequence. Compute $\mb(\hat{x}^n)$, and recalculate the coefficients using \eqref{eq: def of lambda} at $\mb(\hat{x}^n)$. Again, use Viterbi algorithm to solve (P2) at the updated coefficients. Iterate.

For a conditional empirical distribution matrix $\mb$, define its coefficient matrix as $\Lambda(\mb)$, where $\l_{\b,\bb}$ is defined as \eqref{eq: def of lambda}.
For two matrices $A$ and $B$ of the same dimensions, define the scalar product of $A$ and $B$ as
\[
A\odot B \triangleq \sum_{i,j}A_{i,j}B_{i,j}.
\]
Now succinctly, the iterative approach can be described as follows. For $t=0$, let $y^{n,(0)}=x^n$. For $t=1,2,\ldots$
\begin{align*}
\Lambda^{(t)}&=\Lambda(\mb(y^{n,(t-1)})),\nonumber\\
y^{n,(t)}&=\argmin_{z^n\in\hat{\Xc}^n}[\Lambda^{(t)}\odot\mb(z^n)+\a d(x^n,z^n)].
\end{align*}
Stop as soon as $y^{n,(t)} = y^{n,(t-1)}$. 

For a given sequence $x^n$, and slope $\a$, assign to each sequence $y^n\in\hat{\Xc}^n$ the energy
\begin{align}
 \Ec(y^n)=H_k(y^n)+\a d(x^n,y^n).
\end{align}
As mentioned before, the goal is to find the sequence with minimum energy.  Theorem \ref{thm: energy decreases} below gives some justification on how the described  approach serves this purpose. It shows that,  through the iterations,  the energy level of the output is decreasing at each step. Moreover, since the number of energy levels is finite, it proves that the algorithm converges in a finite number of iterations.


\begin{theorem}\label{thm: energy decreases}
For the described iterative algorithm, at each $t\geq1$,
\begin{align}
\Ec(y^{n,(t+1)})\leq \Ec(y^{n,(t)}).
\end{align}
\end{theorem}
\begin{proof}
For the ease of notations, let $\hat{x}^n = y^{n,(t)}$, $\hat{\mb}=\mb(\hat{x}^n)$, and $\hat{\Lambda}=\Lambda(\hat{\mb})$. Similarly, let $\tilde{x}^n=y^{n,(t+1)}$, $\tilde{\mb}=\mb(\tilde{x}^n)$, and $\tilde{\Lambda}=\Lambda(\tilde{\mb})$. From the concavity of $H(\mb)$ in $\mb$,
\begin{align}\label{eq: concavity of H}
H(\tilde{\mb})\leq H(\hat{\mb})+\hat{\Lambda}\odot(\tilde{\mb}-\hat{\mb}),
\end{align}
where $A\odot B$ with  $A$ and $B$ two matrices of the same dimensions is equal to $\sum\limits_{i,j} a_{i,j}b_{i,j}$. On the other hand
\begin{align}\label{eq: equality of H(m)}
\hat{\Lambda}\odot\hat{\mb} &= \sum_{\b,\bb}\hat{\l}_{\b,\bb}\hat{m}_{\b,\bb}\nonumber\\
                      &= \sum_{\b,\bb}\hat{m}_{\b,\bb}\log \left(\frac{\sum\limits_{\b'\in\Xc}\hat{m}_{\b',\bb}}{\hat{m}_{\b,\bb}}\right)\\
                      &= H(\hat{\mb}).
\end{align}
Therefore, combining \eqref{eq: concavity of H} and \eqref{eq: equality of H(m)} yields
\begin{align}\label{eq: ineq of H}
H(\tilde{\mb})\leq \hat{\Lambda}\odot\tilde{\mb}.
\end{align}
Adding a constant term to the both sides of \eqref{eq: ineq of H}, we get
\begin{align}\label{eq: ineq of H}
\Ec(\tilde{x}^n)=H(\tilde{\mb})+\a d(x^n,\tilde{x}^n)\leq \hat{\Lambda}\odot\tilde{\mb}+ \a d(x^n,\tilde{x}^n).
\end{align}
But, since $\tilde{x}^n$ is assumed to be a minimizer of (P2) for the computed coefficients,
\begin{align}\label{eq: ineq of energy}
\hat{\Lambda}\odot\tilde{\mb}+ \a d(x^n,\tilde{x}^n) &\leq \hat{\Lambda}\odot\hat{\mb}+ \a d(x^n,\hat{x}^n)\nonumber\\
                                                     & = H(\hat{\mb})+ \a d(x^n,\hat{x}^n)\nonumber\\
                                                     & = \Ec(\hat{x}^n)
\end{align}
Therefore, combining \eqref{eq: ineq of H} and \eqref{eq: ineq of energy} yields the desired result, i.e.,
\begin{align}
\Ec(\tilde{x}^n) \leq \Ec(\hat{x}^n).
\end{align}
\end{proof}

\begin{remark} \label{remark: alg}

In the described iterative algorithm, for any slope $\a$, we assumed that the algorithm starts at $y^{n,(0)}=x^n$. However, as mentioned earlier, only for large values of $\a$, $\mb(x^n)$ provides a reasonable approximation of the desired type $\mb_n^*$. Hence, in order to address this issue, we can slightly modify the algorithm as follows. The idea is that instead of starting at $y^{n,(0)}=x^n$ for all values of $\a$, we can gradually decrease the slope to our desired value, and use the final output of each step as the initial point for the next step. More explicitly, for any given $\a_0$, start from some large slope, $\a_{\max}$, (corresponding to very low distortion). Run the previous iterative  algorithm  and find $\hat{x}^n(\a_{\max})$. Pick some integer $N_{\a}$, and define
\[
\Delta \a\triangleq{\a_{\max}-\a_0 \over N_{\a}}.
\]

Again run the iterative algorithm, but this time at $\a = \a_{\max}-\Delta \a$. Now, instead of starting from $y^{n,(0)}=x^n$, initialize $y^{n,(0)}=\hat{x}^n(\a_{\max})$. Repeat this process $N_\a$ times. I.e, At the $r^{\rm th}$ step, $r=1,\ldots,N_{\a}$, run the algorithm at $\a = \a_{\max}-r\Delta\a$, and initialize $y^{n,(0)}=\hat{x}^n(\a_{\max}-(r-1)\Delta \a)$. At the final step $\a=\a_0$, and we have a reasonable quantized version of $x^n$ for initialization.\textcolor{white}{aaa}

\end{remark}

To gain further insight on $(P2)$, for the coefficients matrix $\Lambda=\{\l_{\b,\bb}\}_{\b,\bb}$, define
\begin{align}
\phi(\Lambda)&=\min\limits_{y^n\in\hat{\Xc}^n}\left[\sum\limits_{\b,\bb}\l_{\b,\bb}m_{\b,\bb}(y^n)+\a d_n(x^n,y^n)\right]\nonumber\\
        &= \min\limits_{y^n\in\hat{\Xc}^n}\left[\Lambda \odot \mb(y^n)+\a d_n(x^n,y^n)\right].
\end{align}
Since $\phi(\Lambda)$ is the minimum of multiple affine functions of $\La$, it is a concave function. To each sequence $y^n\in\hat{\Xc}^n$, assign a coefficient matrix $\La=[\l_{\b,\bb}]$ as
\begin{align}
\l_{\b,\bb} = \left.\frac{\partial H(\mb)}{\partial m_{\b,\bb}}\right|_{\mb(y^n)}.\label{eq:coeff_vs_m}
\end{align}
Let $\mathcal{L}_d$ be the set of all such coefficient matrices. Similarly to each possible conditional distribution matrix $\mb$ on $\hat{\Xc}^{k+1}$ which satisfies the stationarity condition defined in  Section \ref{subsec:stationary}, assign a coefficients matrix $\La$ defined according to \eqref{eq:coeff_vs_m}.  Let $\mathcal{L}_c$ be the set of coefficient matrices calculated at $(k+1)^{\rm th}$ order {\emph stationary} distributions on $\hat{\Xc}^{k+1}$. Note that while $\mathcal{L}_d$ is a discrete set (consisting of no more than $|\mathcal{Y}|^n$ elements), $\mathcal{L}_c$ is continuous.

For a sequence $x^n$, let
\[
\hat{x}^n = \argmin\limits_{y^n\in\hat{\Xc}^n}\Ec(y^n),
\]
and
\[
\La^*\triangleq\La(\hat{x}^n).
\]
Note that $\La^*$ is the optimal coefficients matrix required for replacing (P1) with (P2).

\begin{lemma}
        \begin{align}
        \La^*&=\argmin\limits_{\Lambda\in\mathcal{L}_d}\phi(\La).\label{eq:lem_main}
        \end{align}
        \end{lemma}
        \begin{proof}
        As shown before,
        \begin{align}
        f(\hat{\La})=\Ec(\hat{x}^n).\label{eq:l1}
        \end{align}
        On the other hand, if $\tilde{x}^n$ is the minimizer of $\sum\limits_{\b,\bb}\l_{\b,\bb}m_{\b,\bb}(y^n)+\a d_n(x^n,y^n)$ for some $\La\in\mathcal{L}_d$, then, as shown in the proof of Theorem \ref{thm: energy decreases},
        \begin{align}
        H(\tilde{\mb})\leq \La \odot \tilde{\mb}.\label{eq:11-2}
        \end{align}
        Therefore,  adding $d(x^n,\tilde{x}^n)$ to both sides of \eqref{eq:11-2} yields
        \begin{align}
        \Ec(\tilde{x}^n) \leq \phi(\La).\label{eq:l2}
        \end{align}
        But, by assumption,
        \begin{align}
        \Ec(\hat{x}^n)\leq\Ec(\tilde{x}^n).\label{eq:l3}
        \end{align}
        Combining \eqref{eq:l1},  \eqref{eq:l2} and  \eqref{eq:l3} yields the desired result.
         \end{proof}

         \begin{remark} Note that
         \begin{align}
         &\min\limits_{\La\in\mathcal{L}_c}\min\limits_{y^n}\left(\ \sum\limits_{\b,\bb}\l_{\b,\bb}m_{\b,\bb}(y^n)+\a d_n(x^n,y^n)\right)\nonumber\\
         &=\min\limits_{y^n}\left[\min\limits_{\La\in\mathcal{L}_c}\left(\ \sum\limits_{\b,\bb}\l_{\b,\bb}m_{\b,\bb}(y^n)\right)+\a d_n(x^n,y^n)\right].
         \end{align}
         But $H(\mb(y^n)) \leq \sum\limits_{\b,\bb}\l_{\b,\bb}m_{\b,\bb}(y^n)$, for any $\La\in\mathcal{L}_c$,  and the lower bound is achieved at $\La(y^n)$. Therefore,
          \begin{align}
         \min\limits_{\La\in\mathcal{L}_c}\min\limits_{y^n}\left(\ \sum\limits_{\b,\bb}\l_{\b,\bb}m_{\b,\bb}(y^n)+\a d_n(x^n,y^n)\right) =\min\limits_{y^n} (H_k(y^n)+\a d(x^n,y^n)).
         \end{align}
         Hence, we can replace $\mathcal{L}_d$ by $\mathcal{L}_c$ in \eqref{eq:lem_main}, and still get the same result. This transform converts the discrete optimization stated in \eqref{eq:lem_main}, which can be solved by exhaustive search, to an optimization over a continuous function of relativley low dimentions.
         \end{remark}

%% file: sim_Viterbi.tex
As the first example, consider an i.i.d.~$\Bern(p)$ source with $p=0.5$. Fig.~\ref{fig:1-1} shows the performance of the iterative algorithm described in Section \ref{sec: viterbi iterative} slightly modified, as suggested in Remark \ref{remark: alg}. The simulations parameters are as follows:
$n= 10^4$, $k=8$, and $\a = (3,2.9,\ldots,0.1)$. Each point corresponds to the average performance over $L=50$ independent source realizations. As mentioned in Section \ref{sec: viterbi iterative}, the iterative algorithm continues until there is no decrease in the cost. Fig.~\ref{fig:1-2} shows the average, minimum and maximum number of required iterations before convergence versus $\alpha$. Again, the number of trials are $L=50$. It can be observed that the number of iterations in this case is always below $60$, which, given the size of the search space, i.e, $2^n$, shows fast convergence. 

\begin{figure}
\begin{center}
\includegraphics[width=13cm]{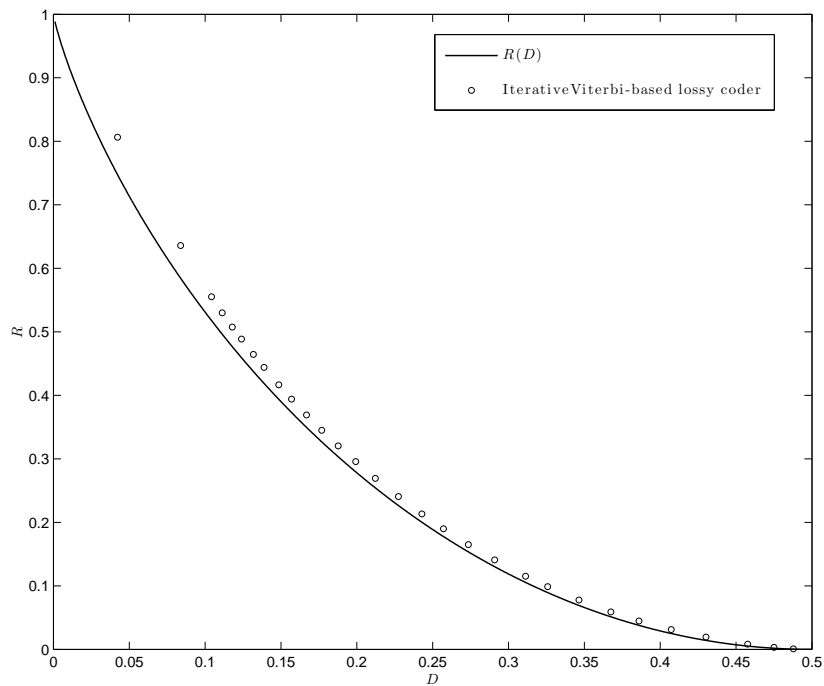}
\caption{Average performance of the iterative Viterbi-based lossy coder applied to an i.i.d.~$\Bern(0.5)$ source. ($n= 10^4$, $k=8$, $\a = (3,2.9,\ldots,0.1)$, and $L=50$)}\label{fig:1-1}
\end{center}
\end{figure}

\begin{figure}
\begin{center}
\includegraphics[width=13cm]{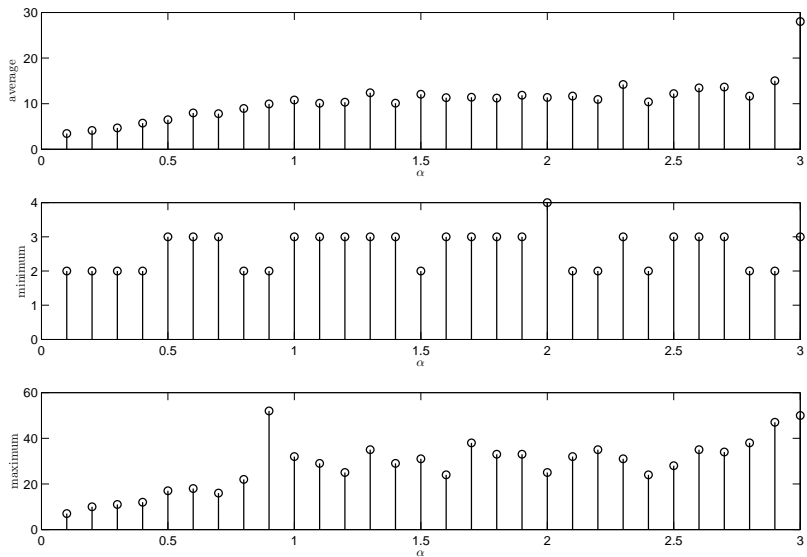}
\caption{From top to bottom: average, minimum and maximum number of iterations before convergence. (i.i.d.~$\Bern(0.5)$ source,  $n= 10^4$, $k=8$, $\a = (3,2.9,\ldots,0.1)$, and $L=50$)}\label{fig:1-2}
\end{center}
\end{figure}

The next example involves a binary symmetric Markov source (BSMS) with transition probability $q=0.2$. Fig.~\ref{fig:2-1} compares the average performance of the Viterbi encoder against upper and lower bounds on $R(D)$ \cite{ISIT07_JW}.   The reason for only comparing the performance of the algorithm against bounds on $R(D)$  in this case is that the rate-distortion function of a Markov source is not known, except for a low-distortion region. For low distortions, the Shannon lower bound is tight \cite{Gray_markov_source}. More explicitly, for $D\leq D_c\approx 0.0159$,
\[
R(D)= H_b(q)-H_b(D),
\]
where $H_b(\epsilon)\triangleq\mathcal{H}(\e,1-\e)$. For $D>D_c$, $R(D)>H_b(q)-H_b(D)$.

A comparison with the memoryless case (Fig.~\ref{fig:1-1}) seems to suggest that the problem is less with how quickly (in $n$) we are converging to the exhaustive search performance scheme of \eqref{eq: exhaustive_search} than with how quickly the convergence in \eqref{eq: achieving optimal point on the rd curve} is taking place, which is source dependent and not at our control.

Fig.~\ref{fig:2-2} shows the average number of iterations before convergence versus $\alpha$. It can be observed that the average is always below $15$.  To give some examples on how the energy is decreasing, Fig.~\ref{fig:2-4} and Fig.~\ref{fig:2-5} show the energy decay through iterations for $\a=1.6$ and $\a=1$ respectively.

\begin{figure}
\begin{center}
\includegraphics[height=11cm,width=15cm]{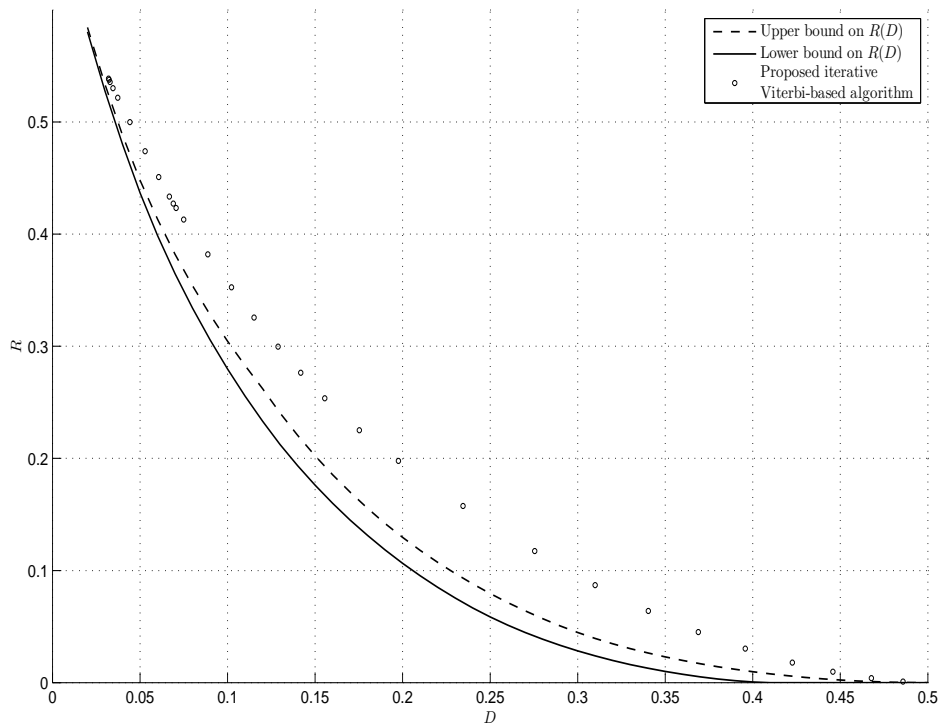}
\caption{Average performance of the iterative Viterbi-based lossy coder applied to a BSMS with $q=0.2$ source. ($n= 25\times10^3$, $k=8$, $\a = 3:-0.1:0.1$ and $L=50$)}\label{fig:2-1}
\end{center}
\end{figure}

\begin{figure}
\begin{center}
\includegraphics[height=11cm,width=15cm]{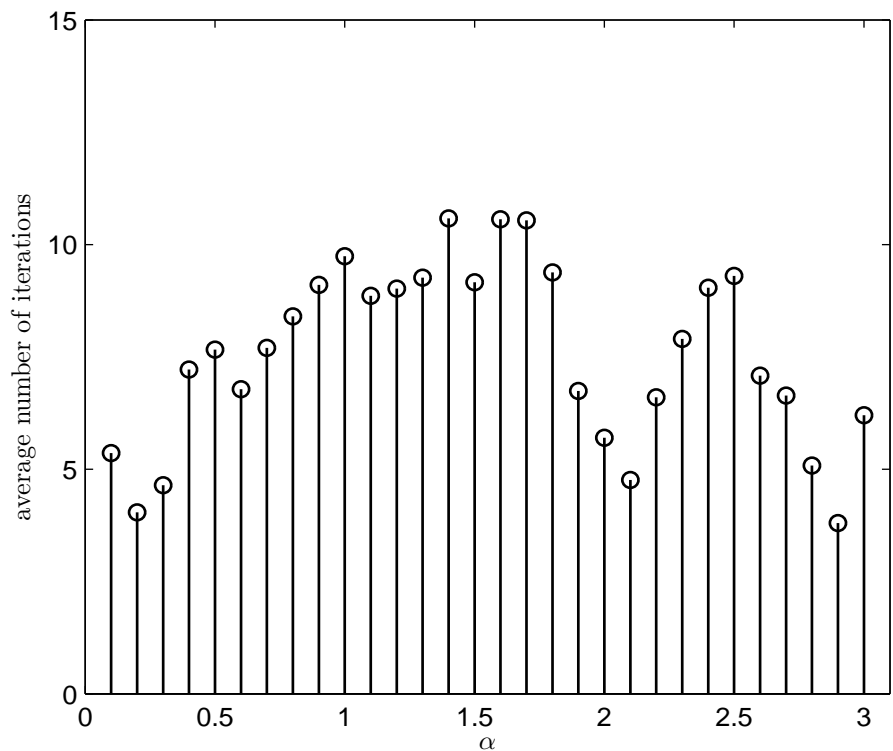}
\caption{Average number of iterations before convergence.(BSMS with $q=0.2$,
 $n= 25\times10^3$, $k=8$, $\a = 3:-0.1:0.1$ and $L=50$)}\label{fig:2-2}
\end{center}
\end{figure}


\begin{figure}
\begin{center}
\includegraphics[height=11cm,width=15cm]{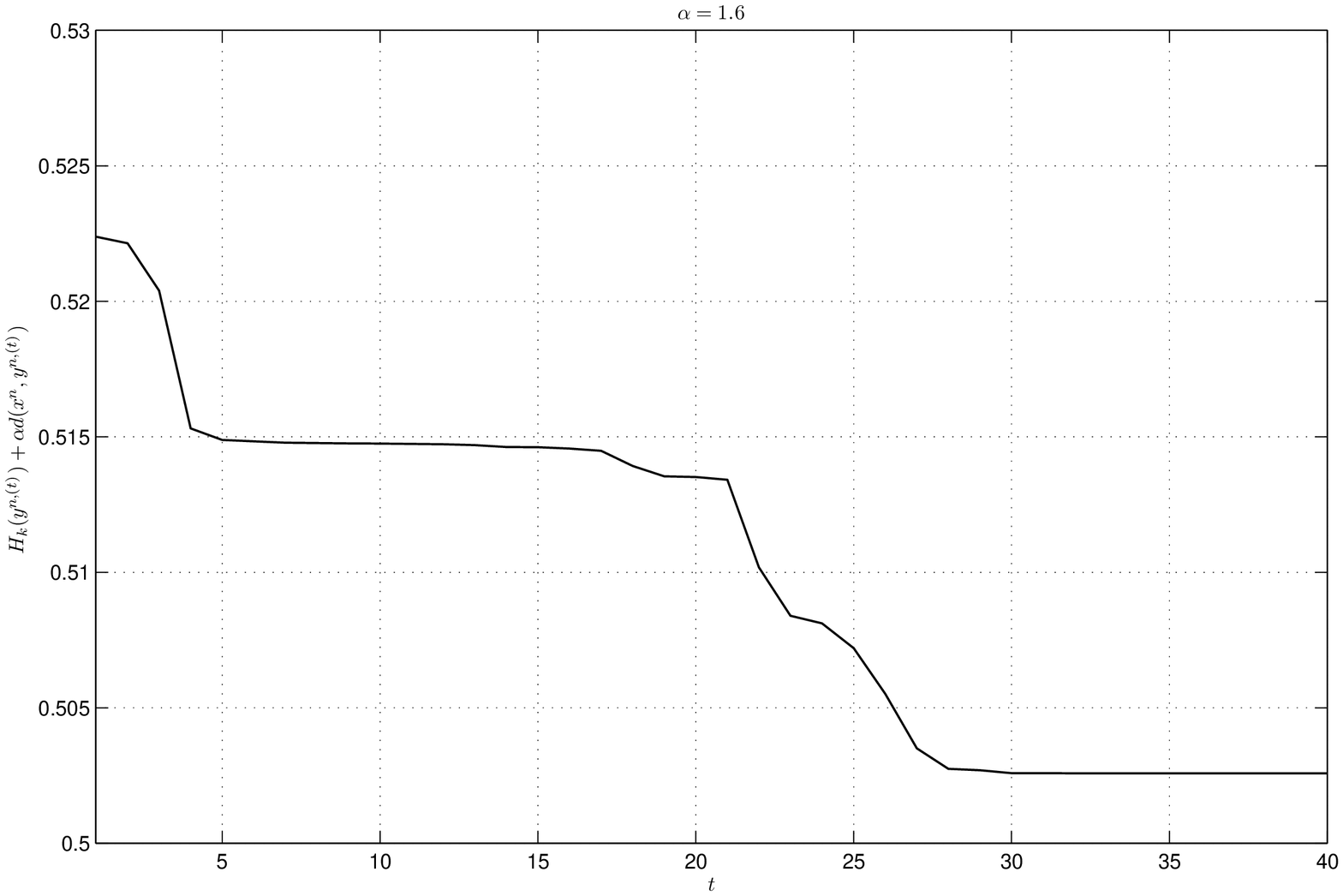}
\caption{Energy decay through the iterations for $\a=1.6$. (BSMS with $q=0.2$,
 $n= 25\times10^3$ and $k=8$)}\label{fig:2-4}
\end{center}
\end{figure}

\begin{figure}
\begin{center}
\includegraphics[height=11cm,width=15cm]{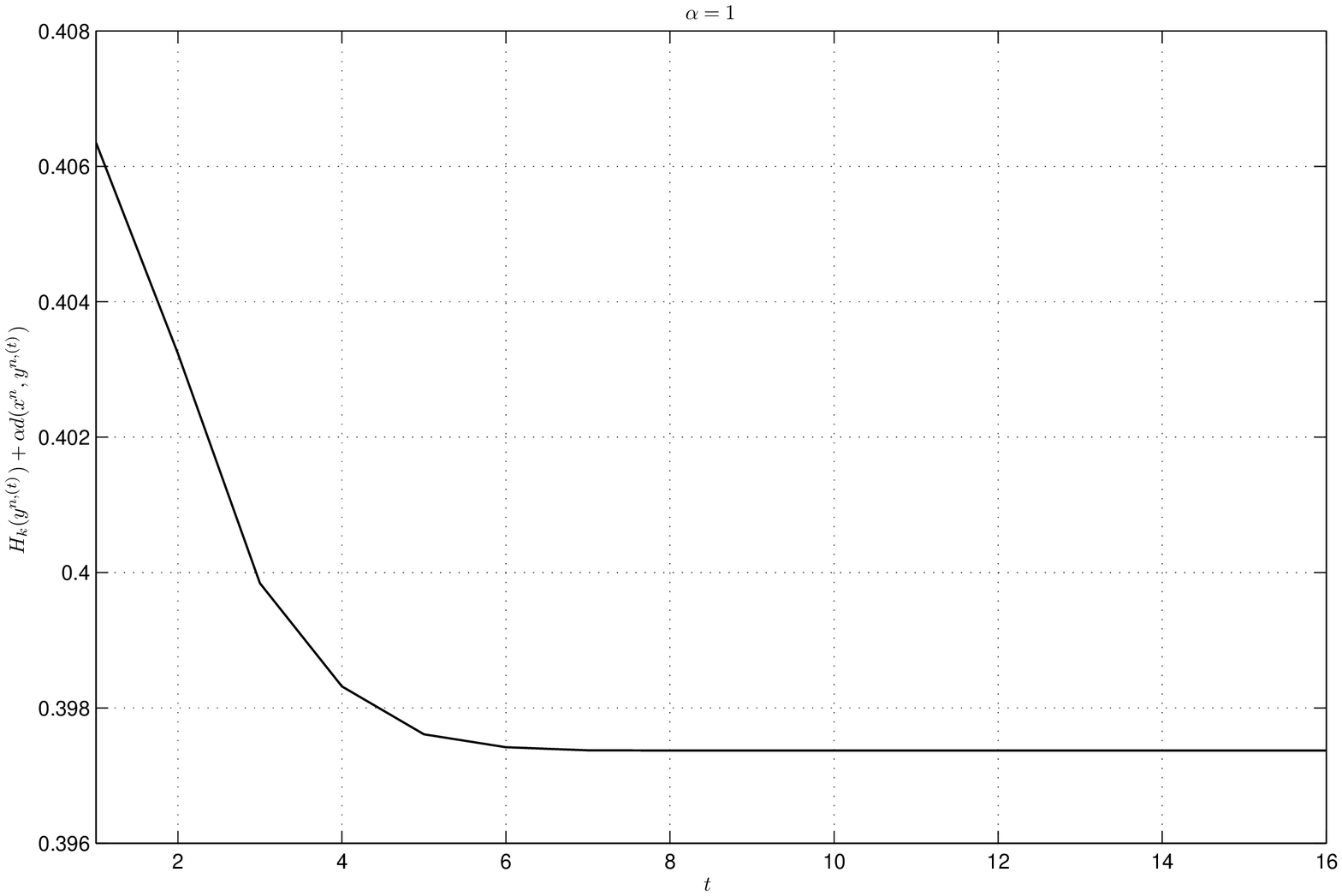}
\caption{Energy decay through the iterations for $\a=1$. (BSMS with $q=0.2$,
 $n= 25\times10^3$ and $k=8$)}\label{fig:2-5}
\end{center}
\end{figure}

\begin{remark}
Similar to \cite{JalaliW:09_arxiv}, here in the figures we are using $H_k(\hat{x}^n)$ as the rate, while in fact it is not a true length function. The reason is that as explained in \cite{JalaliW:09_arxiv}, by Ziv inequality \cite{Ziv_inequality}, if $k=o(\log(n))$, then for any $\e>0$, there exits $N_{\e}\in\mathds{N}$ such that for any $n>N_{\e}$ and any sequence $\mathds{y}=(y_1,y_2,\ldots)$,
\begin{equation} \label{eq: Ziv ineq remark}
 \left[ \frac{1}{n} \ell_{{\sf\footnotesize  LZ}} (y^n) -H_{k}(y^n) \right] \leq \e.
\end{equation}

\end{remark}

%% file: appendix1.tex
\begin{proof}
By rearranging the terms, the cost that is to be minimized in (P1) can alternatively be represented as follows
\begin{align}
H_k(y^n)+\a d_n(x^n,y^n) &= H_k(\mb(y^n))+\a \frac{1}{n}\sum_{i=1}^{n} d(x_i,y_i),\nonumber\\
                         &= H_k(\mb(y^n))+\a \frac{1}{n}\sum_{i=1}^n d(x_i,y_i)\sum\limits_{a\in\Xc,b\in\hat{\Xc}}\ind_{(x_i,y_i)=(a,b)}\nonumber\\
                         &= H_k(\mb(y^n))+\a \frac{1}{n}\sum_{i=1}^n\sum\limits_{a\in\Xc,b\in\hat{\Xc}}d(a,b)\ind_{(x_i,y_i)=(a,b)}\nonumber\\
                         &= H_k(\mb(y^n))+\a \sum\limits_{a\in\Xc,b\in\hat{\Xc}}d(a,b)\frac{1}{n}\sum_{i=1}^n\ind_{(x_i,y_i)=(a,b)}\nonumber\\
                         &= H_k(\mb(y^n))+\a \sum\limits_{a\in\Xc,b\in\hat{\Xc}}d(a,b)\hat{p}^{(1)}_{[x^n,y^n]}(a,b)\nonumber\\
                         &= H_{\hat{p}^{(k+1)}_{[y^n]}}(Y_{k+1}|Y^k)+ \a\E_{\hat{p}^{(1)}_{[x^n,y^n]}}d(X_1,Y_1).
                         \label{eq: new rep of cost}
\end{align}

This new representation reveals the close connection between (P1) and \eqref{eq: original}. Although the costs we are trying to minimize in the two problems are equal, there is a fundamental difference between them: (P1) is a discrete optimization problem, while the optimization space in \eqref{eq: original} is continuous.

Let $\Ec_n^*$ and $\Pc_n^*$ be the sets of minimizers of (P1), and joint empirical distributions of order $\ell$, $\hat{p}_{[x^n,y^n]}^{(\ell)}$, induced by them respectively. Also let $\Sc_n^*$ be the set of marginalized distributions of order $k+1$ in $\Pc_n^*$ with respect to $Y$. Finally, let $C^*_n$ and $\hat{C}^*_n$ be the minimum values achieved by (P1) and \eqref{eq: optimization} respectively.

In order to make the proof more tractable, we break it down into several steps as follows.

\begin{enumerate}

\item Let $y^n\in\Ec_n^*$, and $\hat{p}_{[x^n,y^n]}^{(\ell)}$ be the induced joint empirical distribution. It is easy to check that $\hat{p}_{[x^n,y^n]}^{(\ell)}$ satisfies all  the constraints mentioned in \eqref{eq: optimization}. The only condition that might need some thought is the stationarity constraint, which also holds because
    \begin{align}
        \sum\limits_{a_{\ell}\in\Xc,b_{\ell}\in\hat{\Xc}} \hat{p}_{[x^n,y^n]}^{(\ell)}(a^{\ell},b^{\ell})
        &= \frac{1}{n}\left|\left\{ 1\leq i\leq n: x_{i-\ell+1}^{i-1}=a^{\ell-1}, y_{i-\ell+1}^{i-1}= b^{\ell-1}\right\}\right|,\nonumber\\
        &= \sum\limits_{a_{\ell}\in\Xc,b_{\ell}\in\hat{\Xc}} \hat{p}_{[x^n,y^n]}^{(\ell)}(a_{\ell}a^{{\ell}-1},b_kb^{{\ell}-1}).
    \end{align}
    Therefore, since $\hat{C}_n^*$ is the minimum of \eqref{eq: optimization}, we have
    \begin{align}
    \hat{C}^*_n & \leq H_k(\mb(y^n))+\a \E_{\hat{p}_{[x^n,y^n]}^{(1)}}(X_{k+1},Y_{k+1}) \nonumber\\
    & = H_k(\mb(y^n))+\a d_n(x^n,y^n)\nonumber\\
    & = C_n^*.\label{eq: C_n_s less than C_n}
    \end{align}

\item Let $p^{*(\ell)}\in \hat{\Pc}_n^*$. Based on this joint probability distribution and $x^n$, we construct a reconstruction sequence $\tilde{X}^n$ as follows: divide $x^n$ into $r = \lceil \frac{n}{\ell}\rceil$ consecutive blocks:
    \[
    x^{\ell},x_{\ell+1}^{2\ell},\ldots,x_{(r-2)\ell+1}^{(r-1)\ell},x_{(r-1)\ell+1}^{n},
    \]
    where except for possibly the last block, the other blocks have length $\ell$. The new sequence is constructed as follows
    \[
    \tilde{X}^{\ell},\tilde{X}_{\ell+1}^{2\ell},\ldots,\tilde{X}_{(r-2)\ell+1}^{(r-1)\ell},\tilde{X}_{(r-1)\ell+1}^{n},
    \]
    where for $i=1,\ldots,r-1$, $\tilde{X}_{(i-1)\ell+1}^{i\ell}$ is a sample from the conditional distribution $p^{*(\ell)}(\hat{X}^{\ell}|X^{\ell}=x_{(i-1)\ell+1}^{i\ell})$, and $\tilde{X}_{(r-1)\ell+1}^{n}\sim p^{*(\ell)}(\hat{X}_{(r-1)\ell+1}^{n}|X_{(r-1)\ell+1}^{n}=x_{(r-1)\ell+1}^{n}) $.

\item
    Assume that $\mathbf{x}=\{x_i\}_{i=1}^{\infty}$ is a given individual sequence. For each $n$, let $p^{*(k+1)}$ be the $(k+1)^{\rm th}$ order marginalized version of the solution of \eqref{eq: optimization} on $\hat{\Xc}^{(k+1)}$. Moreover, let $\tilde{X}^n$ be the constructed as described in the previous item, and $\hat{p}_{[\tilde{X}^n]}^{(k+1)}$ be the $(k+1)^{\rm th}$ order empirical distribution induced by $\tilde{X}^n$. We now prove that
    \begin{align}
    \|p^{*(k+1)}-\hat{p}_{[\tilde{X}^n]}^{(k+1)}\|_1\to 0, \;\;{\rm a.s.},\label{eq: as_convergence_p}
    \end{align}
    where the randomization in \eqref{eq: as_convergence_p} is only in the generation of $\tilde{X}^n$.\\
\begin{remark} Since $p^{*(\ell)}$ satisfies \emph{stationarity condition}, its $(k+1)^{\rm th}$ order marginalized distribution, $p^{*(k+1)}$, is well-defined and can be computed with respect to any of the $(k+1)$ consecutive positions in $1,\ldots,\ell$. In other words for $a^{k+1}\in\hat{\Xc}^{k+1}$,
    \begin{align}
    p^{*(k+1)}(a^{k+1}) = \sum\limits_{b^{\ell-k-1}\in\hat{\Xc}^n}p^{*(k+1)}(b^{j}a^{k+1}b_{j+1}^{\ell-k-1}),
    \end{align}
    for any $j\in\{0,\ldots,\ell-k-1\}$, and the result does not depend on the choice of $j$.
\end{remark}

    In order to show that the difference between $\hat{p}_{[\tilde{X}^n]}^{(k+1)}(a^{k+1})$ and $p^{*(k+1)}(a^{k+1})$ is going to zero almost surely, we decompose $\hat{p}_{[\tilde{X}^n]}^{(k+1)}(a^{k+1})$ into the average of $\ell-k$ terms each of which is converging to $p^{*(k+1)}(a^{k+1})$. Then using the union bound we get the desired result which is the convergence of $\hat{p}_{[\tilde{X}^n]}^{(k+1)}(a^{k+1})$ to $p^{*(k+1)}(a^{k+1})$. For $a^{k+1}\in\hat{\Xc}^{k+1}$,
    \begin{align}
    &\left|\hat{p}_{[\tilde{X}^n]}^{(k+1)}(a^{k+1}) - p^{*(k+1)}(a^{k+1})\right| \nonumber\\ &=\left|\frac{1}{n}\sum\limits_{i=1}^{n}\ind_{\tilde{X}_{i-k}^{i}=a^{k+1}} - p^{*(k+1)}(a^{k+1})\right|  \nonumber\\
    &=  \left|\frac{1}{n}\sum\limits_{j=0}^{\ell-k-1}\sum\limits_{i=1}^{r-1}\ind_{\tilde{X}_{i\ell-j-k}^{i\ell-j}=a^{k+1}}+\d_1- p^{*(k+1)}(a^{k+1})\right|, \nonumber\\
    &=  \left|\frac{r}{n}\sum\limits_{j=0}^{\ell-k-1}\left[\frac{1}{r}\sum\limits_{i=1}^{r-1}\ind_{\tilde{X}_{i\ell-j-k}^{i\ell-j}=a^{k+1}}\right]+\d_1- p^{*(k+1)}(a^{k+1})\right|, \nonumber\\
    &=  \left|\frac{1}{\ell-k}\sum\limits_{j=0}^{\ell-k-1}\left[\frac{1}{r}\sum\limits_{i=1}^{r-1}\ind_{\tilde{X}_{i\ell-j-k}^{i\ell-j}=a^{k+1}}\right]+\d_2- p^{*(k+1)}(a^{k+1})\right|, \label{eq: p_hat_k}
    \end{align}
    where $\d_1$ accounts for the edge effects between the blocks, and $\d_2$ is defined such that $\d_2-\d_1$ takes care of the effect of replacing $\frac{r}{n}$ with $\frac{1}{\ell-k}$. Therefore, $0\leq \d_1< {(k+1)r \over n} + {\ell-1 \over n}\leq  {2(k+1) \over \ell} + {1 \over r}$, and $ |\d_2-\d_1| = o(k/{\ell})$. Hence, $\d_1\to 0$ and $\d_2\to 0$ as $n\to\infty$. 

The new representation decomposes a sequence of correlated random variables, $\{\ind_{\tilde{X}_{i-k}^i=a^{k+1}}\}_{i=k+1}^n$, into $\ell-k$ sub-sequences where each of them is an independent process. For achieving this some counts that lie between two blocks are ignored, i.e., if $\ind_{\tilde{X}_{i-k}^i=a^{k+1}}$ is such that it depends on more than one block of the form $\tilde{X}_{(i-1)\ell+1}^{i\ell}$, we ignore it. The effect of such ignored counts will be no more than $\d_r$ which goes to zero as $k,\ell\to\infty$ because the theorem requires $k=o(\ell)$. More specifically in \eqref{eq: p_hat_k}, for each $j\in\{0,\ldots,\ell-k-1\}$, $\{\ind_{\tilde{X}_{i\ell-j-k}^{i\ell-j}=a^{k+1}}\}_{i=1}^r$ is a sequence of independent  not necessarily identically distributed random variables.

    For $n$ large enough, $|\d_2|<\e/2$. Therefore, by Hoeffding inequality \cite{Hoeffding}, and the union bound,
    \begin{align}
    &\P\left(\left|\hat{p}_{[\tilde{X}^n]}^{(k+1)}(a^{k+1}) - p^{*(k+1)}(a^{k+1})\right|>\e \right),\nonumber\\ &\leq\P\left(\left|\frac{1}{\ell-k}\sum\limits_{j=0}^{\ell-k-1}\left[\frac{1}{r}\sum\limits_{i=1}^{r-1}\ind_{\tilde{X}_{i\ell-j-k}^{i\ell-j}=a^{k+1}}- p^{*(k+1)}(a^{k+1})\right]\right|>\frac{\e}{2}\right),\nonumber\\
    &\leq\P\left(\frac{1}{\ell-k}\sum\limits_{j=0}^{\ell-k-1}\left|\frac{1}{r}\sum\limits_{i=1}^{r-1}\ind_{\tilde{X}_{i\ell-j-k}^{i\ell-j}=a^{k+1}}- p^{*(k+1)}(a^{k+1})\right|>\frac{\e}{2}\right),\nonumber\\
    &\leq\sum\limits_{j=0}^{\ell-k-1}\P\left(\left|\frac{1}{r}\sum\limits_{i=1}^{r-1}\ind_{\tilde{X}_{i\ell-j-k}^{i\ell-j}=a^{k+1}}- p^{*(k+1)}(a^{k+1})\right|>\frac{\e}{2}\right),\nonumber\\
    & \leq 2(\ell-k)e^{-r\e^2/2}.\label{eq: delta p_hat and p_star}
    \end{align}

Again by the union bound,
    \begin{align}
    &\P\left(\|\hat{p}^{(k+1)}_{[\tilde{X}^n]}-p^{*(k+1)}\|_1> \e \right)\nonumber\\
    &\hspace{1cm}\leq \sum\limits_{a^{k+1}\in\hat{\Xc}^{k+1}}\P\left(\left|\hat{p}_{[\tilde{X}^n]}^{(k+1)}(a^{k+1}) - p^{*(k+1)}(a^{k+1})\right|>\frac{\e}{|\hat{\Xc}|^{k+1}} \right),\nonumber\\
    &\hspace{1cm}\leq |\hat{\Xc}|^{k+1} 2(\ell-k)e^{-\frac{n\e^2}{2\ell|\hat{\Xc}|^{2(k+1)}}}. \label{eq: bound_tv_emp}
    \end{align}
    Our choices of $k = k_n = o(\log n)$, $\ell = \ell_n = o(n^{1/4})$, $k=o(\ell)$, and $k_n, \ell_n\to\infty$, as $n\to\infty$ now guarantee that the right hand side of \eqref{eq: bound_tv_emp} is summable on $n$ which together with Borel-Cantelli Lemma yields the desired result of \eqref{eq: as_convergence_p}.

    \item Using similar steps as above we can prove that
    \begin{align}
    \|q^*-\hat{q}^{(1)}_{[x^n,\tilde{X}^n]}\|\to 0,\;\;{\rm a.s.}
    \end{align}
    Again we first prove that $|q^*(a,b)-\hat{q}^{(1)}_{[x^n,\tilde{X}^n]}(a,b)|\to 0$ for each $a\in\Xc$ and $b\in \hat{\Xc}$. For doing this we again need to decompose
    \[
    \{\ind_{x_i=a,\tilde{X}_i=b}\}_{i=1}^n
    \]
    into $\ell$ sub-sequences each of which is a sequence of independent random variables, and then apply Hoeffding inequality plus the union bound. Finally we apply the union bound again in addition to the Borel-Cantelli Lemma to get our desired result.

    \item Combing the results of the last two parts, and the fact that $H_k(\mb)$ and $\E_qd(X,Y)$ are bounded continuous functions of $\mb$ and $q$ respectively, we conclude that
    \begin{align}
    H_k(\tilde{X}^n)+\a d_n(x^n,\tilde{X}^n) &= H_{\hat{p}^{(k+1)}_{[\tilde{X}^n]}}(Y_{k+1}|Y^k)+ \a\E_{\hat{q}^{(1)}_{[x^n,\tilde{X}^n]}}d(X_1,Y_1)\nonumber\\
    &= H_{{p}^{*(k+1)}}(Y_{k+1}|Y^k)+ \a\E_{q^{*}}d(X_1,Y_1)+\e_n\nonumber\\
    &=\hat{C}^*_n + \e_n, \label{eq: approx_cost_by_X_t}
    \end{align}
    where $\e_n\to 0$ with probability $1$.

    \item Since $C_n^*$ is the minimum of (P1), we have
    \begin{align}
    C_n^* &\leq H_k(\tilde{X}^n)+\a d_n(x^n,\tilde{X}^n),\nonumber\\
    &= \hat{C}^*_n+\e_n.
    \end{align}
     On the other hand, as shown in \eqref{eq: C_n_s less than C_n}, $\hat{C}^*_n\leq C_n^*$. Therefore,
    \begin{align}
    |C_n^*-\hat{C}^*_n|\to 0 \label{eq: diff C_n C_hat_n}
    \end{align}
    as $n\to \infty$.

    \item For a given set of coefficients $\boldsymbol\l=\{\l_{\b,\bb}\}_{\b,\bb}$ computed at some $\mb$ according to \eqref{eq: def of lambda}, define
        \begin{align}
        f(\boldsymbol\l)=\min\limits_{y^n\in\hat{\Xc}^n}\left[\sum\limits_{\b,\bb}\l_{\b,\bb}m_{\b,\bb}(y^n)+\a d_n(x^n,y^n)\right].
        \end{align}
        It is easy to check that $f$ is continuous, and bounded by $1+\a$. Therefore, since $\boldsymbol\l$ is in turns a continuous function of $\mb$, and as proved in \eqref{eq: as_convergence_p},
        \[
        \|p^{*(k+1)}-\hat{p}_{[\tilde{X}^n]}^{(k+1)}\|_1\to 0,
        \]
        we conclude that,
        \begin{equation}
        |f(\boldsymbol\l^*)-f(\hat{\boldsymbol\l})|\to 0,\label{eq: diff_f}
      \end{equation}
      where $\boldsymbol\l^*$ and $\hat{\boldsymbol\l}$  are the coefficients computed at $p^{*(k+1)}$ and $\hat{p}_{[\tilde{X}^n]}^{(k+1)}$ respectively.

    \item Let $\bar{X}^n$ be the output of (P2) when the coefficients are computed at $\mb(\tilde{X}^n)$. Then, from Theorem \ref{thm: energy decreases},
      \begin{align}
        H_{k}(\bar{X}^n) +\a  d_n (x^n, \bar{X}^n) &\leq H_{k}(\tilde{X}^n) +\a  d_n (x^n, \tilde{X}^n)\nonumber\\
                                                   &= \hat{C}_n^*+{\e}_n.
      \end{align}
      Since, $\e_n\to 0$, this shows that haven computed the coefficients at $\mb(\tilde{X}^n)$, we would get a universal lossy compressor. But instead, we want to compute the coefficients at $\mb^*$. From \eqref{eq: diff_f}, the difference between the performances of these two algorithms goes to zero. Therefore, we finally get our desired result which is
      \begin{equation}
         \left[ H_{k}(\hat{X}^n) +\a  d_n (X^n, \hat{X}^n ) \right] \stackrel{n \rightarrow \infty}{\longrightarrow}  \min_{D \geq 0} \left[ R(\mathbf{X}, D)
        +\a D \right],\;\;{\rm a.s.}
      \end{equation}

\end{enumerate}
\end{proof}

%% file: it_Viterbi_v3.bbl
\begin{thebibliography}{10}

\bibitem{cover}
T.~Cover and J.~Thomas.
\newblock {\em Elements of Information Theory}.
\newblock Wiley, New York, 2nd edition, 2006.

\bibitem{Shannon60}
C.~Shannon.
\newblock Coding theorems for a discrete source with fidelity criterion.
\newblock In R.~Machol, editor, {\em Information and Decision Processes}, pages
  93--126. McGraw-Hill, 1960.

\bibitem{Gallager}
R.G. Gallager.
\newblock {\em Information Theory and Reliable Communication}.
\newblock NY: John Wiley, 1968.

\bibitem{book:Berger}
T.~Berger.
\newblock {\em Rate-distortion theory: A mathematical basis for data
  compression}.
\newblock NJ: Prentice-Hall, 1971.

\bibitem{YangZ:97}
En~hui Yang, Z.~Zhang, and T.~Berger.
\newblock Fixed-slope universal lossy data compression.
\newblock {\em Information Theory, IEEE Transactions on}, 43(5):1465--1476, Sep
  1997.

\bibitem{Sakrison:70}
D.~J. Sakrison.
\newblock The rate of a class of random processes.
\newblock {\em Information Theory, IEEE Transactions on}, 16:10--16, Jan. 1970.

\bibitem{Ziv:72}
J.~Ziv.
\newblock Coding of sources with unknown statistics part ii: Distortion
  relative to a fidelity criterion.
\newblock {\em Information Theory, IEEE Transactions on}, 18:389--394, May
  1972.

\bibitem{NeuhoffG:75}
D.~L. Neuhoff, R.~M. Gray, and L.D. Davisson.
\newblock Fixed rate universal block source coding with a fidelity criterion.
\newblock {\em Information Theory, IEEE Transactions on}, 21:511--523, May
  1972.

\bibitem{NeuhoffS:78}
D.~L. Neuhoff and P.~L. Shields.
\newblock Fixed-rate universal codes for {Markov} sources.
\newblock {\em Information Theory, IEEE Transactions on}, 24:360--367, May
  1978.

\bibitem{Ziv:80}
J.~Ziv.
\newblock Distortion-rate theory for individual sequences.
\newblock {\em Information Theory, IEEE Transactions on}, 24:137--143, Jan.
  1980.

\bibitem{GarciaN:82}
R.~Garcia-Munoz and D.~L. Neuhoff.
\newblock Strong universal source coding subject to a rate-distortion
  constraint.
\newblock {\em Information Theory, IEEE Transactions on}, 28:285Ð295, Mar.
  1982.

\bibitem{LZ}
J.~Ziv and A.~Lempel.
\newblock Compression of individual sequences via variable-rate coding.
\newblock {\em Information Theory, IEEE Transactions on}, 24(5):530--536, Sep
  1978.

\bibitem{arith_coding}
I.~H. Witten, R.~M. Neal, , and J.~G. Cleary.
\newblock Arithmetic coding for data compression.
\newblock {\em Commun. Assoc. Comp. Mach.}, 30(6):520--540, 1987.

\bibitem{CheungW:90}
K.~Cheung and V.~K. Wei.
\newblock A locally adaptive source coding scheme.
\newblock {\em Proc. Bilkent Conf on New Trends in Communication, Control, and
  Signal Processing}, pages 1473--1482, 1990.

\bibitem{MoritaK:89}
H.~Morita and K.~Kobayashi.
\newblock An extension of {LZW} coding algorithm to source coding subject to a
  fidelity criterion.
\newblock In {\em In Proc. 4th Joint Swedish-Soviet Int. Workshop on
  Information Theory}, page 105–109, Gotland, Sweden, 1989.

\bibitem{SteinbergG:93}
Y.~Steinberg and M.~Gutman.
\newblock An algorithm for source coding subject to a fidelity criterion based
  on string matching.
\newblock {\em Information Theory, IEEE Transactions on}, 39:877Ð886, Mar.
  1993.

\bibitem{YangK:98}
En~hui Yang and J.C. Kieffer.
\newblock On the performance of data compression algorithms based upon string
  matching.
\newblock {\em Information Theory, IEEE Transactions on}, 44(1):47 --65, jan
  1998.

\bibitem{LuczakS:97}
T.~Luczak and T.~Szpankowski.
\newblock A suboptimal lossy data compression based on approximate pattern
  matching.
\newblock {\em Information Theory, IEEE Transactions on}, 43:1439Ð1451, Sep.
  1997.

\bibitem{ZamirR:01}
R.~Zamir and K.~Rose.
\newblock Natural type selection in adaptive lossy compression.
\newblock {\em Information Theory, IEEE Transactions on}, 47(1):99 --111, jan
  2001.

\bibitem{AtallahG:99}
W.~Szpankowski Atallah, Y.~G\'enin.
\newblock Pattern matching image compression: algorithmic and empirical
  results.
\newblock {\em IEEE Trans. Pattern Analysis and Machine Intelligence},
  21:618Ð627, Sept. 1999.

\bibitem{DemboK:99}
Amir Dembo and Ioannis Kontoyiannis.
\newblock The asymptotics of waiting times between stationary processes,
  allowing distortion.
\newblock {\em The Annals of Applied Probability}, 9(2):413--429, May 1999.

\bibitem{MarcellinF:90}
M.~W. Marcellin and T.~Fischer.
\newblock Trellis coded quantization of memoryless and {Gauss-Markov} sources.
\newblock {\em IEEE Trans. on Comm.}, 38(1):82--93, jan 1990.

\bibitem{BergerG:98}
T.~Berger and J.D. Gibson.
\newblock Lossy source coding.
\newblock {\em Information Theory, IEEE Transactions on}, 44(6):2690--2723, Sep
  1998.

\bibitem{book:GershoGray:92}
A.~Gersho and R.M. Gray.
\newblock {\em Vector Quantization and Signal Compression}.
\newblock Springer, New York, 1992.

\bibitem{KasnerM:99}
J.H. Kasner, M.W. Marcellin, and B.R. Hunt.
\newblock Universal trellis coded quantization.
\newblock {\em Image Processing, IEEE Transactions on}, 8(12):1677 --1687, dec
  1999.

\bibitem{WainrightM:05}
M.J. Wainwright and E.~Maneva.
\newblock Lossy source encoding via message-passing and decimation over
  generalized codewords of {LDGM} codes.
\newblock In {\em Proc. IEEE Int. Symp. Inform. Theory}, pages 1493--1497,
  Sept. 2005.

\bibitem{GuptaV:09}
A.~Gupta and S.~Verd\'u.
\newblock Nonlinear sparse-graph codes for lossy compression.
\newblock {\em Information Theory, IEEE Transactions on}, 55(5):1961 --1975,
  may 2009.

\bibitem{GuptaV:08}
A.~Gupta, S.~S.~Verd\'u, and T.~Weissman.
\newblock Rate-distortion in near-linear time.
\newblock In {\em Proc. IEEE Int. Symp. Inform. Theory}, pages 847--851,
  Toronto, Canada, July 2008.

\bibitem{Arikan:09_arxiv}
E.~Arikan.
\newblock Channel polarization: A method for constructing capacity-achieving
  codes for symmetric binary-input memoryless channels.
\newblock {\em arXiv:0807.3917}.

\bibitem{BabuR:09_arxiv}
S.~Babu~Korada and R.~Urbanke.
\newblock Polar codes are optimal for lossy source coding.
\newblock {\em arXiv:0903.0307}.

\bibitem{JalaliW:09_arxiv}
S.~Jalali and T.~Weissman.
\newblock Rate-distortion via {Markov} chain {Monte Carlo}.
\newblock {\em arXiv:0808.4156v2}.

\bibitem{GrayN:75}
R.~Gray, D.~Neuhoff, and J.~Omura.
\newblock Process definitions of distortion-rate functions and source coding
  theorems.
\newblock {\em Information Theory, IEEE Transactions on}, 21(5):524--532, Sep
  1975.

\bibitem{Viterbi:67}
A.~Viterbi.
\newblock Error bounds for convolutional codes and an asymptotically optimum
  decoding algorithm.
\newblock {\em Information Theory, IEEE Transactions on}, 13(2):260 -- 269, apr
  1967.

\bibitem{Forney:73}
Jr. Forney, G.D.
\newblock The {Viterbi} algorithm.
\newblock {\em Proceedings of the IEEE}, 61(3):268 -- 278, march 1973.

\bibitem{ISIT07_JW}
S.~Jalali and T.~Weissman.
\newblock New bounds on the rate-distortion function of a binary {Markov}
  source.
\newblock In {\em Proc. IEEE Int. Symp. Inform. Theory}, Nice, France, July
  2007.

\bibitem{Gray_markov_source}
R.~Gray.
\newblock Rate distortion functions for finite-state finite-alphabet markov
  sources.
\newblock {\em Information Theory, IEEE Transactions on}, 17(2):127--134, Mar
  1971.

\bibitem{Ziv_inequality}
E.~Plotnik, M.J. Weinberger, and J.~Ziv.
\newblock Upper bounds on the probability of sequences emitted by finite-state
  sources and on the redundancy of the {Lempel-Ziv} algorithm.
\newblock {\em Information Theory, IEEE Transactions on}, 38(1):66--72, Jan
  1992.

\bibitem{Hoeffding}
W.~Hoeffding.
\newblock Probability inequalities for sums of bounded random vaiables.
\newblock {\em Journal of the American Statistical Association},
  58(301):13--30, March 1963.

\end{thebibliography}
